\documentclass[12pt,a4paper]{article}
\usepackage{amsmath,amsthm,amsfonts,amssymb,bbm}
\usepackage{graphicx,psfrag,subfigure,color}
\usepackage{cite}
\usepackage{hyperref}

\numberwithin{equation}{section}

\newcommand{\const}{\mathrm{const}}
\newcommand{\Pb}{\mathrm{Prob}}
\newcommand{\Prob}{\mathbb{P}}

\newcommand{\Id}{\mathbbm{1}}

\newcommand{\I}{{\rm i}}
\newcommand{\D}{\mathrm{d}}

\newcommand{\R}{\mathbb{R}}
\newcommand{\N}{\mathbb{N}}
\newcommand{\Z}{\mathbb{Z}}

\DeclareMathOperator*{\virt}{virt}

\newtheorem{prop}{Proposition}[section]
\newtheorem{thm}[prop]{Theorem}
\newtheorem{lem}[prop]{Lemma}
\newtheorem{defin}[prop]{Definition}
\newtheorem{cor}[prop]{Corollary}

\newtheorem{cla}[prop]{Claim}

\newtheorem{rem}[prop]{Remark}

\title{Random tilings and Markov chains for interlacing particles}
\date{June 12, 2015}
\author{Alexei Borodin\thanks{Massachusetts Institute of Technology, Department of Mathematics, 77 Massachusetts Avenue, Cambridge, MA 02139-4307, USA. E-mail: {\tt borodin@math.mit.edu}}\and Patrik L.\ Ferrari\thanks{Institute for Applied Mathematics, Bonn University, Endenicher Allee 60, 53115 Bonn, Germany. E-mail: {\tt ferrari@uni-bonn.de}}}

\begin{document}

\maketitle
\sloppy
\begin{abstract}
We explain the relation between certain random tiling models and interacting
particle systems belonging to the anisotropic KPZ (Kardar-Parisi-Zhang)
universality class in $2+1$-dimensions. The link between these two \emph{a
priori} disjoint sets of models is a consequence of the presence of shuffling
algorithms that generate random tilings under consideration. To see the
precise connection, we represent both a random tiling and the corresponding
particle system through a set of non-intersecting lines, whose dynamics is
induced by the shuffling algorithm or the particle dynamics. The resulting class
of measures on line ensembles also fits into the framework of the Schur
processes.
\end{abstract}

\newpage
\section{Introduction}\label{SectIntro}
In this contribution we start by considering the domino tilings of the Aztec
diamond, probably the most well-known and studied random tiling model, and we
explain in quite some details how this model is equivalent to an interacting
particle system. After that two generalizations are presented.

The connection between the Aztec diamond and a $(1+1)$-dimensional interacting
particle system is well-known since the work of Jockusch, Propp and
Shor~\cite{JPS98}, where they used it to prove the \emph{arctic circle theorem}.
Indeed, one quarter of the arctic circle is nothing else than the hydrodynamic
limit of the height function associated with a totally asymmetric simple
exclusion process in discrete time with parallel update and packed initial
condition.

In~\cite{BF08} we constructed a class of $(2+1)$-dimensional interacting particle
systems, some of which have the property that their fixed time projections can
be viewed as a random tiling model, while certain $(1+1)$-dimensional projections
represent models in the Kardar-Parisi-Zhang universality class of stochastic
growth models (like totally asymmetric simple exclusion processes). The domino
tilings of the Aztec diamond is a special case that fits in that framework as
remarked at the end of Section 2.6 of~\cite{BF08}. This was not a complete
surprise, as the connection between the shuffling algorithm of the Aztec diamond
and an interacting particle system in $2+1$ dimensions was previously obtained
by Nordenstam~\cite{Nor08} (this was later extended to a few other cases
in~\cite{NY11}).

The particle configurations in the state space of the (discrete time) Markov
chain in~\cite{Nor08} are not in a one-to-one mapping to domino tilings of the
Aztec diamond, rather in a one-to-many relation. Here \emph{many} is not an
arbitrary quantity, but it is the number of random fillings in the shuffling
algorithm during the last iteration. The mapping is the same as the \mbox{one-to-many}
mapping between tilings of the Aztec diamond and the six-vertex model with
domain wall boundary conditions at the free-fermion point~\cite{ZJ00,FS06}. In
order to have a bijection one needs to have the information of the particle
configurations at the two latest time moments (this fact is used for instance
in the java applet~\cite{FerAZTEC}). The reason for this fact will be clear
when discussing the parallel update in the Markov chain,
see Section~\ref{SectAKPZ}.

To explain the relation between the Aztec diamond and our interacting particle
system, we map both systems to a set of non-intersecting (broken) lines viewed
as paths on a Lindstr\"om-Gessel-Viennot (LGV) planar directed graph. They
were first introduced by Johansson in~\cite{Jo03} in the study of the Airy$_2$
fluctuations of the border of the disordered phase. The previously studied line
ensembles representations as in~\cite{ZJ00} are still one-to-many, while the one
used in~\cite{Jo03,FS06} is in bijection with the tiling configurations.
This bijection results in the shuffling algorithm~\cite{EKLP92,Pro03} inducing
a Markovian dynamics on the ensembles of lines. The same line ensembles with
the same dynamics can be also obtained from the configurations of a system of
interacting particles in \mbox{$2+1$-dimensions} in discrete time with parallel
update~\cite{BF08} and with a particular initial condition.

The probability measures (on tilings, sets of particles, or line
ensembles) one obtains in this way  fall into the class of so-called
conditional $L$-ensembles and thus have determinantal correlation
functions~\cite{RB04}. These measures and the corresponding Markov dynamics can
also be described in terms of dimer models on Rail
Yard graphs~\cite{BBBCCV14,BBCCR15}.

\subsubsection*{Connection with the Schur processes}
The constructions of the present paper can be recast into the framework of the
Schur processes (introduced in~\cite{OR01}) and Schur dynamics on them
(introduced in~\cite{Bor10}). While we won't do that here (see, however,
a recent exposition in~\cite{BBBCCV14}), let us make the connection more
explicit. The basic fact is that the weights that we associate to the two
sections of our LGV graphs
\begin{figure}
\begin{center}
\psfrag{m1}[r]{$\mu_1-1$}
\psfrag{m2}[r]{$\mu_2-2$}
\psfrag{m3}[r]{$\mu_3-3$}
\psfrag{l1}[c]{$\lambda_1-1$}
\psfrag{l2}[c]{$\lambda_2-2$}
\psfrag{l3}[c]{$\lambda_3-3$}
\psfrag{n1}[l]{$\nu_1-1$}
\psfrag{n2}[l]{$\nu_2-2$}
\psfrag{n3}[l]{$\nu_3-3$}
\includegraphics[height=4cm]{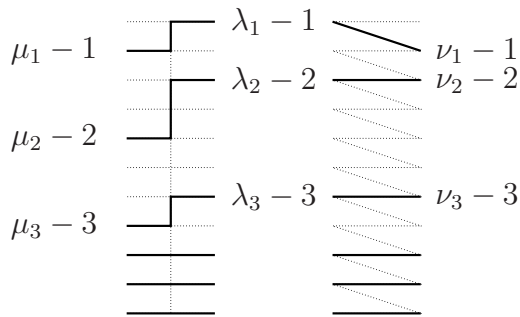}
\caption{LGV graphs and Schur process. The weight of the left one is $s_{\lambda/\mu}(\alpha)$, the weight of the right one is $s_{\lambda/\nu}(\hat \beta)$.}
\label{FigSchur}
\end{center}
\end{figure}
with weights of vertical edges in the first graph being $\alpha$ and weights of
skew edges in the second graph being $\beta$ (all other edges have weight 1),
can be viewed as values of the skew Schur functions
$s_{\lambda/\mu}(\alpha)$ and $s_{\lambda/\nu}(\hat \beta)$. Here $\alpha$
stands for the specialization of the symmetric functions into a single non-zero
variable equal to $\alpha$ (with complete homogeneous symmetric functions
specializing into $h_n(\alpha)=\alpha^n$, $n\geq 0$), and $\hat\beta$ stands for
the specialization into a single 'dual' variable equal to $\beta$ (with
complete homogeneous symmetric functions specializing into $h_n(\hat\beta)=1$ if
$n=0$, $h_n(\hat\beta)=\beta$ if $n=1$, and $h_n(\hat\beta)=0$ if $n\geq 2$).
The commutativity relations of type (\ref{eqIntertwining}) below then turn into
instances of the skew-Cauchy identity (see e.g.~\cite{Bor10}).

We chose to use more general notation in this paper (and in~\cite{BF08})
because of a potential applicability of the general setup to other models, which
was indeed realized for example in~\cite{BC11}.

Most of the work reported in this note was done in 2010, but it seems to not have
lost all of its novelty yet. Indeed very recently we have been asked about the details
of the connection between our work~\cite{BF08} and random tiling models. Thus
we decided to make them available with this publication.

\subsubsection*{Outlook}
In Section~\ref{SectAztec} we define the Aztec diamond, its domino
tilings, and introduce the mapping to non-intersecting line ensembles. In
Section~\ref{SectAKPZ} we present a class of Markov dynamics on interlacing
particle systems and explain which initial conditions and parameters one has
to choose to generate the same measure as the non-intersecting line ensembles
of the Aztec diamond. Finally, in Section~\ref{SectGeneralization} we present
two generalizations of the particle dynamics and explain which random tilings
they correspond to.

\subsubsection*{Acknowledgments}
The authors are grateful for earlier discussions with Senya Shlosman and more
recently with Dan Betea, J\'er\'emie Bouttier, and Sunil Chhita. The work of A.~Borodin is supported by the NSF grant
DMS-1056390. The work of P.L.~Ferrari is supported by the German Research
Foundation via the SFB 1060--B04 project.

\section{The Aztec diamond}\label{SectAztec}
In this section we introduce the Aztec diamond, a well-known random model, and explain his representation as a system of non-intersecting lines. The latter will be used later to make the connection with a Markov chain dynamics on interlacing particle configurations.

\subsection{The random tiling model}
Domino tilings of the Aztec diamond were introduced in~\cite{EKLP92}. For any \mbox{$N\in\N$}, we define the Aztec diamond $A_N$ of size $N$ as the union of all lattice squares $[m,m+1]\times[n,n+1]$, $m,n\in\Z$, such that they are inside the region \mbox{$\{(x,y): |x|+|y|\leq N+1\}$}. A tiling of the Aztec diamond is a configurations of $N^2$ dominoes (i.e., $2\times 1$ rectangles) such that the Aztec diamond $A_N$ is fully covered by dominoes. Alternatively, one can think of a tiling as a perfect matching on the bipartite graph obtained by considering the dual graph inside the Aztec diamond, see Figure~\ref{FigAztec1}.
\begin{figure}
\begin{center}
\includegraphics[height=5cm]{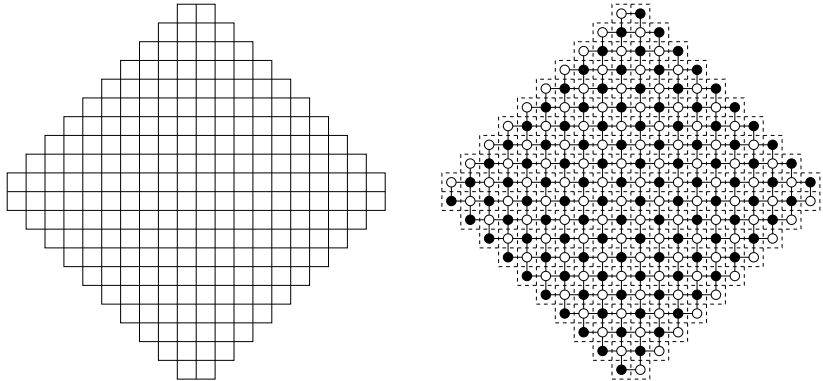}
\caption{Left panel: The Aztec diamond of size $N=10$. Right panel: The associated bipartite graph.}
\label{FigAztec1}
\end{center}
\end{figure}
Denote by ${\cal T}_N$ the set of all possible tilings of $A_N$. It is well-known~\cite{EKLP92} that $|{\cal T}_N|=2^{N(N+1)/2}$. A random tiling of the Aztec diamond $A_N$ is obtained defined by choosing a probability measure on ${\cal T}_N$. Specifically we consider the probability measure $\Prob_N$ given by
\begin{equation}\label{eqMeasAztec}
\Prob_N(T)=\frac{a^{v(T)}}{\sum_{S\in {\cal T}_N}a^{v(S)}},\quad T\in {\cal T}_N,
\end{equation}
where $v(S)$ is the number of vertical dominoes in $S$, and $a>0$ is a parameter.

As one can see from Figure~\ref{FigAztec2}, a typical random tiling of the Aztec diamond has four ordered regions and one disordered region in the middle, that has been studied in great details: the law of large number (Artic circle phenomenon)~\cite{JPS98}, the fluctuations of the boundaries (Airy$_2$ process)~\cite{Jo03}, local statistics~\cite{CEP96}, and Gaussian Free Field fluctuations in the bulk\footnote{The detailed computations of the convergence to the GFF have not been written down in papers. The precise statement can be found in Section 6 of~\cite{CJY14}. They are essentially the same computations as the ones made in Section 5 of~\cite{BF08}, except that the kernel is the one associated with the discrete time TASEP with parallel update~\cite{BFS07b}. At the end of the argument, one uses ideas that goes back to~\cite{Ken01}. A detailed statement of the expected result is available in Section 6 of~\cite{CJY14}.}, and the GUE-minor process at the turning points~\cite{JN06}.
\begin{figure}
\begin{center}
\includegraphics[height=4.5cm,angle=45]{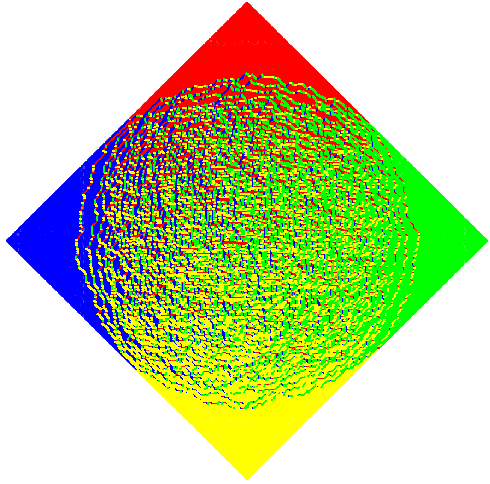}
\caption{A random tiling of an Aztec diamond of size $N=200$ and $a=1$. The colors represents the four types of dominoes.}
\label{FigAztec2}
\end{center}
\end{figure}

\subsubsection*{The shuffling algorithm}
To generate a random tiling of the Aztec diamond of size $N$, one can proceed iteratively from the Aztec diamond of size $N-1$ by means of the \emph{shuffling algorithm} introduced in~\cite{EKLP92}, that we recall briefly. For that purpose, let us give some names to the different types of dominoes: by superimposing the Aztec diamond to a checkerboard table (which of the chosen two possible ways is irrelevant), there are actually four types of dominoes, as the black site might be either to the right/left (resp. top/bottom) site of a horizontal (resp. vertical) domino. We call them North, East, South, and West dominoes according to the following rule: the North (resp. South) domino is the horizontal one that fits into the top (resp. bottom) most part of the Aztec diamond, and similarly the East (resp. West) domino is the vertical one that fits into the right (resp. left) most part of the Aztec diamond.

The shuffling algorithm is the following. Start with a domino tiling of $A_N$ distributed according to $\Prob_N$:
\begin{itemize}
\item[Step 1:] Move all the dominoes in the direction of their names by one unit. If in doing it a North and a South domino (resp. a East and a West domino) exchange their positions, remove them.
\item[Step 2:] At this point, the dominoes partially tile $A_{N+1}$. The empty region can be uniquely decomposed into $2\times 2$ blocs. Independently of each other, each bloc is tiled with two horizontal dominoes with probability $1/(1+a^2)$ or two vertical dominoes with probability $a^2/(1+a^2)$.
\end{itemize}
This procedure generates a random tiling of the Aztec diamond $A_{N+1}$ with distribution $\Prob_{N+1}$ as proven in~\cite{EKLP92}.

\subsection{Line ensembles for the Aztec diamond}
To each tiling of the Aztec diamond one can associate a set of non-intersecting lines bijectively as shown by Johansson in~\cite{Jo03}. We use a slightly different but equivalent representation. Add horizontal lines in the middle of the South-type dominoes, a 45-degrees oriented segment in the West-type dominoes and a step-down in the East-type dominoes as indicated in Figure~\ref{FigAztec3}. For a tiling of $A_N$ this results in $N$ non-intersecting lines with fixed initial and end points.
\begin{figure}
\begin{center}
\includegraphics[height=5.5cm]{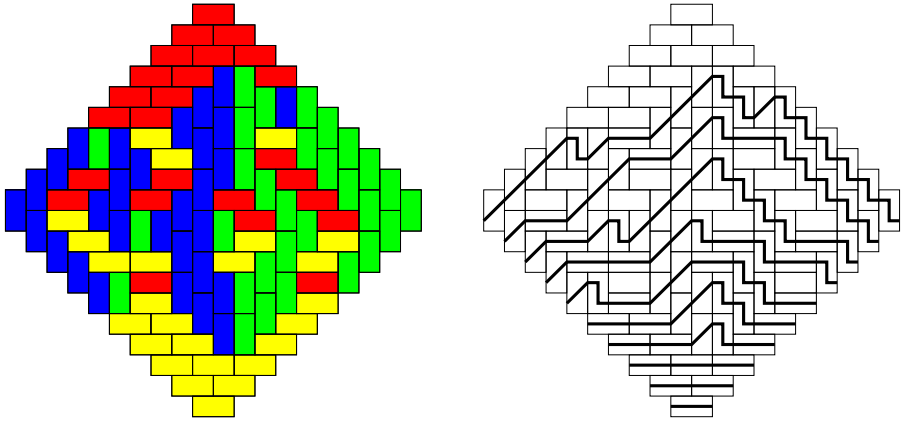}
\caption{A random tiling of an Aztec diamond of size $N=10$ and its associates set of lines. The North dominoes are red, the South are yellow, the East are green, and finally the West dominoes are blue.}
\label{FigAztec3}
\end{center}
\end{figure}

Further, one can think of the Aztec diamond of size $N$ as being embedded into tilings of $\R^2$, where outside the Aztec diamond we add only horizontal dominoes that do not overlap and fill the whole space, see the gray tiles in Figure~\ref{FigAztec4} (left).
\begin{figure}
\begin{center}
\includegraphics[height=5cm]{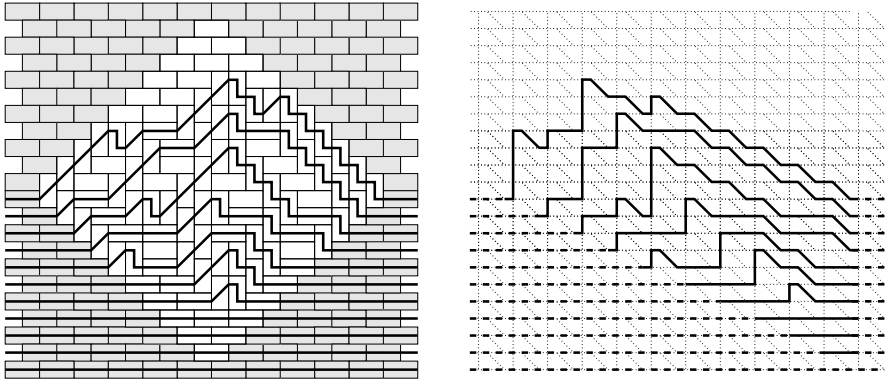}
\caption{A random tiling of an Aztec diamond of size $N=10$. Left panel: The original tiling of Figure~\ref{FigAztec3} with the lines extended deterministically outside $A_N$. Right pane: Final representation of the ensemble of non-intersecting lines. The dotted lines is the underlying LGV graph. The lines which are fixed and correspond to the gray dominoes of the left picture are dashed.}
\label{FigAztec4}
\end{center}
\end{figure}
By a simple geometric transformation, which can be easily recovered by comparing the left and the right illustrations in Figure~\ref{FigAztec4}, we obtain the final non-intersecting line ensemble representation of the random tiling of Figure~\ref{FigAztec3}.

\vspace{0.3em}
\noindent\begin{minipage}{0.83\textwidth}
\hspace{1em} The above procedure gives a bijection between a tiling of an Aztec diamond and a configuration of non-intersecting lines on a Lindstr\"om-Gessel-Viennot (LGV) directed graph. The basic building bloc of the LGV graph is the one here on the right. Then, if we consider an Aztec diamond of size $N$, we have a one-to-one bijection between the set of non-intersecting lines on the LGV graph that consists of $N$ copies of the basic building bloc and where the lines start and end at all vertical positions $-1,-2,\ldots$, see Figure~\ref{FigAztec4} for an example.
\end{minipage}
\hfill
\begin{minipage}{0.15\textwidth}
\begin{center}
\includegraphics[height=4cm]{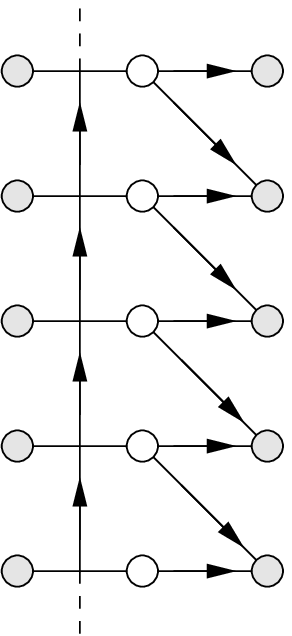}
\end{center}\vfill
\end{minipage}
\vspace{0.3em}

Now we have a bijection between a set of non-intersecting line ensembles and the set of all possible tilings of the Aztec diamond $A_N$. The next question is to know which weights we have to assign to the edges of the LGV graph in order to obtain the distribution on $A_N$ given by (\ref{eqMeasAztec}), i.e., where each vertical domino has weight $a$ and each horizontal domino has weight $1$. It is easy to see that it is enough to give weight $1$ to each horizontal edge of the LGV graph and weight $a$ to each of the vertical edges and the down-right edges. Alternatively, since the number of East and West dominoes are equal, we can also give weight $1$ to each horizontal edge, weight $\alpha$ to each vertical edge and $\beta$ to each down-right edge, provided that $\alpha \beta=a^2$.

\subsubsection*{Induced line ensembles dynamics}
Because of the bijection described above, the shuffling algorithm induces a Markov dynamics on the set of non-intersecting line ensembles. To see how the dynamics works, it might be useful to decorate the line ensembles with dominoes of the same colors as in Figure~\ref{FigAztec3}, see Figure~\ref{FigAztec5}.
\begin{figure}
\begin{center}
\psfrag{t}[c]{\small $t$}
\psfrag{0}[c]{\small $0$}
\psfrag{1}[c]{\small $1$}
\psfrag{2}[c]{\small $2$}
\psfrag{3}[c]{\small $3$}
\psfrag{4}[c]{\small $4$}
\includegraphics[height=6cm]{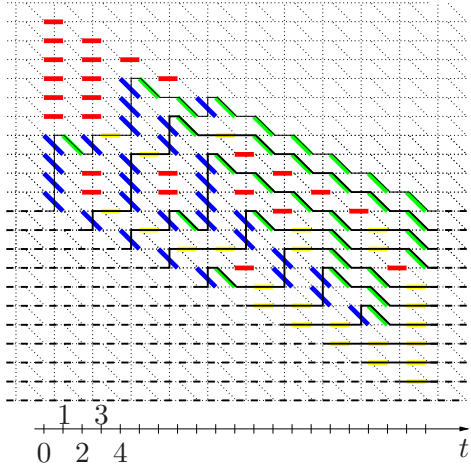}
\caption{Line ensembles decorated with domino colors as in Figure~\ref{FigAztec3}.}
\label{FigAztec5}
\end{center}
\end{figure}
The configuration at time $n=0$ consists only of straight lines at height $\{-1,-2,\ldots\}$. Given a line ensemble corresponding to an Aztec tiling of size $n$, the line ensemble of a random tiling of $A_{n+1}$ induced by the shuffling algorithm is the following (see Figure~\ref{FigShuffling} for an example):
\begin{itemize}
\item[Step 1:] Update the lines at the \emph{even} lines-times (=horizontal coordinates) $t$: their height (=vertical coordinates) is updated to be equal to the height at lines-time $t-1$. This is because by Step 1 of the shuffling algorithm, in the representation of Figure~\ref{FigAztec5}, the blue dominoes stay put, the red dominoes move vertically by $2$, the green dominoes moves to the right by $1$ and vertically by $1$, the yellow dominoes moves to the right by $1$, and when the trajectories of a red and a yellow (resp.\ green and blue) dominoes meet, they are deleted.
\item[Step 2:] Update the lines at the \emph{odd} lines-times $t\in\{1,3,\ldots,2n+1\}$: if a line can be extended vertically by one unit without generating intersection, then this happens with probability $a^2/(1+a^2)$ (it corresponds to adding a blue/green domino, the vertical ones in the Aztec representation), while the move is not made with probability $1/(1+a^2)$. All the possible updates in this step are independent of each other.
\end{itemize}

\begin{figure}
\begin{center}
\psfrag{T0}[c]{Example of lines for a random tiling of $A_1$}
\psfrag{Step1}[c]{After Step 1}
\psfrag{Step2}[c]{After Step 2}
\includegraphics[height=16cm]{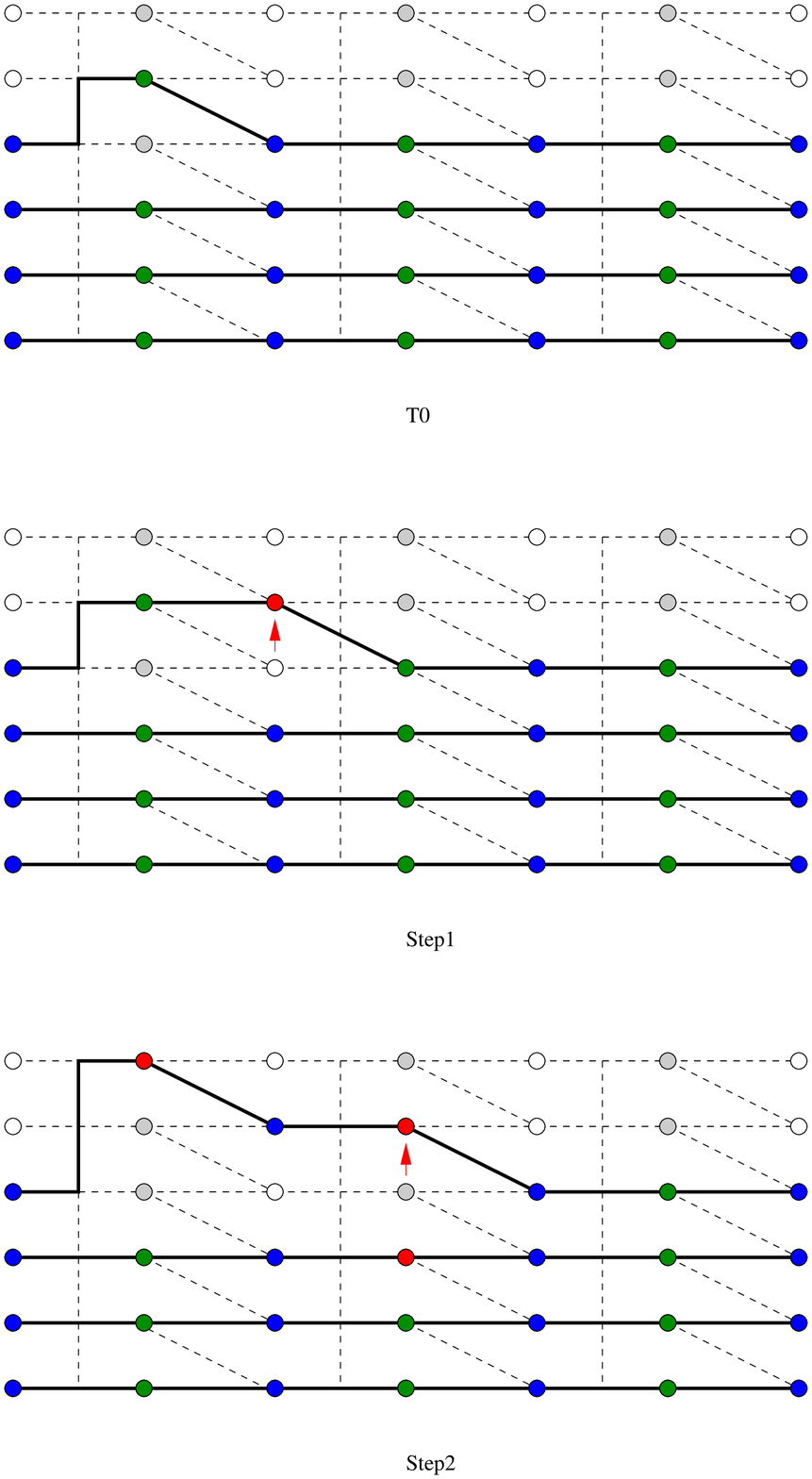}
\caption{An example of the dynamics on the non-intersecting lines. The red dots are the points which, due to the shuffling dynamics, have to move (further indicated by a small arrow) or can more. After \emph{Step 1} the set of lines does not overlaps with the LGV graphs. This is the reason why the top line in \emph{Step 2} has to further move up.}
\label{FigShuffling}
\end{center}
\end{figure}

\newpage
\section{Markov chains on interlacing particle}\label{SectAKPZ}
In this section we first introduce a class of discrete time Markov chains on interlacing particle systems. It is a special case of the more general framework developed in~\cite{BF08}. With a particular choice of transition matrices and initial conditions, we can use the particle dynamics to generate non-intersecting line ensembles with the distributions equal to the ones induced by the shuffling algorithm. For pedagogical reasons, we focus here on the situation corresponding to the Aztec diamond. Two generalizations will be presented in Section~\ref{SectGeneralization}.

\subsection{Construction of the Markov chain}

\subsubsection*{The state space}
The state space of our Markov chain consists of interlacing configurations of particles, also known as \emph{Gelfand-Tsetlin patterns}, defined as follows:
\begin{equation}
\small \mathbb{GT}_N=\left\{X^N=(x^1,\ldots,x^N), x^n=(x_1^n,\ldots,x_n^n)\in\Z^n\, |\, x^n\prec x^{n+1}, 1\leq n\leq N-1\right\},
\end{equation}
where
\begin{equation}\label{eq3.2}
x^n\prec x^{n+1}\iff x_1^{n+1} < x_1^n \leq x_2^{n+1} < x_2^{n+1}\leq \ldots < x_n^n \leq x_{n+1}^{n+1}.
\end{equation}
If $x^n\prec x^{n+1}$ we say that \emph{$x^n$ interlaces with $x^{n+1}$}. We can (and will) think of the configurations as unions of levels: so the vector $x^n$ of $X^N$ is the state at \emph{level $n$}. See Figure~\ref{figGTN} for an illustration.
\begin{figure}
\begin{center}
  \psfrag{m}[c]{$<$} \psfrag{e}[c]{$\leq$} \psfrag{x11}{$x_1^1$}
  \psfrag{x12}{$x_1^2$} \psfrag{x13}{$x_1^3$} \psfrag{x14}{$x_1^4$}
  \psfrag{x22}{$x_2^2$} \psfrag{x23}{$x_2^3$} \psfrag{x24}{$x_2^4$}
  \psfrag{x33}{$x_3^3$} \psfrag{x34}{$x_3^4$} \psfrag{x44}{$x_4^4$}
  \includegraphics[height=3.5cm]{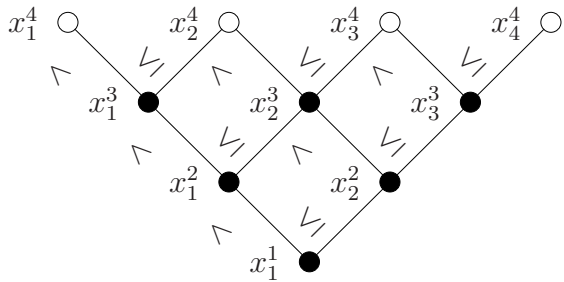}
\caption{Graphical representation of $\mathbb{GT}_4$. The white dots represents the vector at level $n=4$.}
\label{figGTN}
\end{center}
\end{figure}

The Markov chain is built up through two basic Markov chains: (a) the first one is the time evolution at a fixed level and (b) the second one is a Markov chain on $\mathbb{GT}_N$, linking level $n$ with level $n-1$, $2\leq n\leq N$. We first discuss the two chains separately and then define the full Markov chain.

\subsubsection*{Markov chain at a fixed level}
The dynamics at a fixed level is a discrete time analogue of the famous Dyson's Brownian Motion that appears in random matrices. Consider the one-particle transition matrix with entries\footnote{For a set $S$, by $\frac{1}{2\pi\I}\oint_{\Gamma_S} \D w f(w)$ we mean the contour integral where the contour can be taken to be any anticlockwise oriented simple path containing all the points in the set $S$ but no other singularities of the function $f$.}
\begin{equation}\label{eq2.11}
{\bf P}(x,y)= \frac{1}{2\pi\I}\oint_{\Gamma_0} \D w \frac{1-p+p w^{-1}}{w^{x-y+1}}
=\left\{
         \begin{array}{ll}
           p, & \textrm{if }y=x+1, \\
           1-p, & \textrm{if }y=x, \\
           0, & \textrm{otherwise}.
         \end{array}
       \right.
\end{equation}
Then, the one-particle transition probability is given by
\begin{equation}
p_t(x)=({\bf P}^t)(0,x)=\binom{t}{x}p^x (1-p)^{t-x}\Id_{[0\leq x\leq t]}
\end{equation}
for $t\in\N$.

For any fixed $n\in \{1,\ldots,N\}$, let us define an $n$-particle process on the projection of the state space $\mathbb{GT}_N$ to the level-$n$, that is the
Weyl chamber $W_n$,
\begin{equation}
W_n=\{x=(x_1,\ldots,x_n)\in\Z^n\, | \, x_1<x_2<\ldots<x_n\}.
\end{equation}
Define also the boundary of the Weyl-chamber by
\begin{equation}
\partial W_n=\{x=(x_1,\ldots,x_n)\in\Z^n\, | \, x_i=x_{i+1} \textrm{ for some }i\in \{1,\ldots,n-1\}\}´.
\end{equation}
Consider first the process on \mbox{$x^n=(x_1^n,\ldots,x_n^n)\in\Z^n$} where each of the $x_k^n$, $1\leq k \leq n$, are independent processes with transition matrix $\bf P$ given above. The (discrete time) generator of the random walk is
\begin{equation}
(L_n f)(x)=\sum_{i=1}^n p(f(x+e_i)-f(x)),
\end{equation}
with $e_i$ the vector with entries $e_i(j)=\delta_{i,j}$. We want to condition the random walk $x^n$ on never having a collision between any two of its components, i.e., to stay in the Weyl chamber forever. The conditioning can be achieved by the Doob $h$-transform with $h$ given by the Vandermonde determinant, see e.g.~\cite{OCon02}. Indeed, one can verify that
\begin{equation}\label{eq2.5}
h_n(x)=\prod_{1\leq i < j \leq n} (x_j-x_i)\equiv \Delta_n(x)
\end{equation}
is harmonic and vanishes at $\partial W_n$, i.e., $L_n h_n = 0$. Then, the conditioned process given that the walker never leaves the Weyl chamber $W_n$ is the Doob $h$-transform of the free walk. More precisely, for $x,y\in W_n$ and $t\in\N$, the transition probability $P_{n,t}$ from $x$ to $y$ of the conditioned random walk is given by
\begin{equation}
P_{n,t}(x,y)=\frac{h_n(y)}{h_n(x)}\Pb(x^n(t)=y, T>t\,|\, x^n(0)=x),
\end{equation}
where $T=\min\{t\geq\,0 |\, x^n(t)\not\in W_n\}$. Since our random walks are one-sided, $x^n$ can exit the Weyl-chamber only through $\partial W_n$. Therefore using Karlin-McGregor's type formula (or, since we are in discrete time, the LGV theorem), we have
\begin{equation}
\Pb(x^n(t)=y, T>t\,|\, x^n(0)=x) = \det\left(p_t(y_i-x_j)\right)_{i,j=1}^n.
\end{equation}
Consequently, we define the level-$n$ chain as follows.
\begin{defin}\label{DefDiscreteCharlier}
The transition probability of the Markov chain at level $n$ is given by
\begin{equation}\label{eq2.14}
P_{n}(x,y)=\frac{\Delta_n(y)}{\Delta_n(x)}\det\left({\bf P}(x_i,y_j)\right)_{i,j=1}^n
\end{equation}
with $\bf P$ as in (\ref{eq2.11}), $x,y\in W_n$, and $\Delta_n$ the Vandermonde determinant.
\end{defin}

\subsubsection*{Markov chain at a fixed time}
Now we consider a Markov link between levels of $\mathbb{GT}_N$ that generates the uniform measure on $\mathbb{GT}_N$ given the value $x^N$ of the level $N$, i.e., with \mbox{$x^N=(x_1^N<x_2^N<\ldots<x_N^N)$} fixed. It is well-known, see e.g.\ Corollary~\ref{CorA4}, that
\begin{equation}\label{eq2.12}
\textrm{\# of $\mathbb{GT}_N$ patterns with given $x^N$}=\prod_{1\leq i < j \leq N}\frac{x^N_j-x_i^N}{j-i} = \frac{\Delta_N(x^N)}{\prod_{n=1}^{N-1}n!}.
\end{equation}
Thus, the uniform measure on $\mathbb{GT}_N$ given $x^N$ can be obtained by setting
\begin{equation}
\Pb(x^{N-1}\, | \, x^N)=\frac{\textrm{\# of $GT_{N-1}$ patterns with given $x^{N-1}$}}{\textrm{\# of $\mathbb{GT}_N$ patterns with given $x^N$}}\Id_{[x^{N-1}\prec x^N]}.
\end{equation}
Using (\ref{eq2.12}) we obtain
\begin{equation}
\Pb(x^{N-1}\, | \, x^N) = (N-1)! \frac{\Delta_{N-1}(x^{N-1})}{\Delta_N(x^N)}\Id_{[x^{N-1}\prec x^N]}.
\end{equation}
Consequently, we define the following.
\begin{defin}\label{DefMarkovLink}
For any $n\in\{2,\ldots,N\}$, define the Markov link between level $n$ and $n-1$ by
\begin{equation}\label{eqMarkovLink}
\Lambda^n_{n-1}(x^n,x^{n-1}):=(n-1)!\frac{\Delta_{n-1}(x^{n-1})}{\Delta_n(x^n)}\Id_{[x^{n-1}\prec x^n]},
\end{equation}
where $x^n\in W_n$ and $x^{n-1}\in W_{n-1}$.
\end{defin}
The uniform measure on $\mathbb{GT}_N$ given $x^N$ can be expressed by
\begin{equation}\label{eq2.13}
\frac{\prod_{n=1}^{N-1}n!}{\Delta_N(x^N)} \Id_{[x^1\prec x^2\prec\ldots\prec x^N]} =\prod_{n=2}^N \Lambda^n_{n-1}(x^n,x^{n-1}).
\end{equation}

There is an important representation of the interlacing through determinants. This is relevant when doing concrete computations.
\begin{lem}\label{LemInterlacing}
Let $x^N\in W_N$ and $x^{N-1}\in W_{N-1}$ be ordered configurations in the Weyl chambers. Then, setting $x_N^{N-1}\equiv {\rm virt}$ a ``virtual variable'', we have
\begin{equation}\label{eqInterlacingDeterminants}
\Id_{[x^{N-1}\prec x^N]} = \det(\phi(x_j^N,x_i^{N-1}))_{i,j=1}^N,
\end{equation}
with $\phi(x,y)=\Id_{[y>x]}=\frac{1}{2\pi\I}\oint_{\Gamma_0} \D w \frac{(1-w)^{-1}}{w^{y-x}}$ and $\phi(x,{\rm virt})=1$.
\end{lem}
\begin{proof}
The proof is quite easy. For $x^{N-1}\prec x^N$ one sees that the matrix on the r.h.s.~of (\ref{eqInterlacingDeterminants}) is triangular with $1$ on the diagonal. Further, by violating the interlacing conditions, one gets two rows or columns which are equal, so that the determinant is equal to zero.
\end{proof}

\subsubsection{The intertwining condition}
The key property used in our dynamics on $\mathbb{GT}_N$ is the \emph{intertwining relation}, illustrated in Figure~\ref{figIntertwining}, namely
\begin{equation}\label{eqIntertwining}
\Delta^{n}_{{n-1}}:=P_{n} \Lambda^n_{n-1} = \Lambda^n_{n-1} P_{n-1},\quad 2\leq n \leq N.
\end{equation}
In our specific case, to see that (\ref{eqIntertwining}) holds one uses the Fourier representation for $P_n$ (see (\ref{eq2.11})) and of $\Lambda^n_{n-1}$ (see Lemma~\ref{LemInterlacing}). The intertwining relation can be obtained quite generically when the transition matrices are translation invariant, see Proposition 2.10 of~\cite{BF08} or Appendix~\ref{AppToeplitz} for more details. Intertwining Markov chains were first studied in~\cite{DF90}.
\begin{figure}
\begin{center}
  \psfrag{L}[l]{$\Lambda^{n}_{n-1}$}
  \psfrag{Pn}[c]{$P_n$}
  \psfrag{Pnm1}[c]{$P_{n-1}$}
  \psfrag{Sn}[c]{$W_n$}
  \psfrag{Snm1}[c]{$W_{n-1}$}
  \psfrag{C}[c]{$\circlearrowleft$}
  \includegraphics[height=3.5cm]{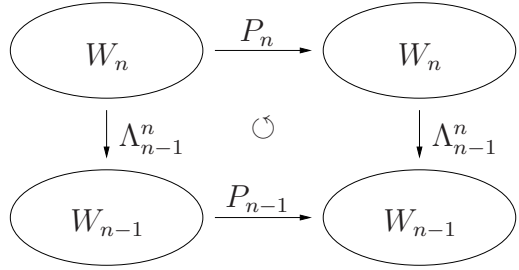}
\caption{Intertwining of the Markov chains.}
\label{figIntertwining}
\end{center}
\end{figure}

\subsubsection{The sequential update dynamics}
Now we define a discrete time Markov chain on $\mathbb{GT}_N$ that we call \emph{sequential update}, since the update occurs level by level and for the update of level $n$ we use the updated values of level $n-1$. In particular, the projection to the particles $\{x_1^1,x_1^2,\ldots\}$ is the sequential update dynamics of the totally asymmetric simple exclusion process (TASEP).

Define the transition probabilities of a Markov chain on $\mathbb{GT}_N$ by (we use the notation $X^N(t)=(x^1(t),\ldots,x^N(t))$)
\begin{multline}\label{eqMConGTn}
P_{\Lambda}^N(X^N(t),X^N(t+1))
=P_{1}(x^1(t),x^1(t+1))\\
\times \prod_{k=2}^N\frac{P_{k}(x^k(t),x^k(t+1))\Lambda^k_{k-1}(x^k(t+1),x^{k-1}(t+1))}{\Delta^k_{k-1}(x^k(t),x^{k-1}(t+1))},
\end{multline}
for $X^N(t),X^N(t+1)\in \mathbb{GT}_N$, see Figure~\ref{FigSequential} for an illustration.

One can think of $P_{\Lambda}^N$ as follows. Starting from $X^N(t)=(x^1(t),\ldots,x^N(t))$, we first choose $x^1(t+1)$ according to the transition matrix \mbox{$P_{1}(x^1(t),x^1(t+1))$}, then choose $x^2(t+1)$ using $\frac{P_{2}(x^2(t),x^2(t+1))\Lambda^2_{1}(x^2(t+1),x^{1}(t+1))}{\Delta^2_{1}(x^2(t),x^{1}(t+1))}$, which is the conditional distribution of the middle point in the successive application of $P_{2}$ and $\Lambda^2_1$, provided that we start at $x^2(t)$ and finish at $x^1(t+1)$. After that we choose $x^3(t+1)$ using the conditional distribution of the middle point in the successive application of $P_{3}$ and $\Lambda^3_2$ provided that we start at $x^3(t)$ and finish at $x^2(t+1)$, and so on.

With our choice of $P_n$'s and $\Lambda^n_{n-1}$'s introduced above, the dynamics is the following:
\begin{itemize}
  \item $x_1^1$ just performs a one-sided random walk (with jump probability~$p$).
  \item $x_1^2$ performs a one-sided random walk but the jumps leading to $x_1^2=x_1^1$ are suppressed (we say that $x_1^2$ is blocked by $x_1^1$).
  \item $x_2^2$ performs a one-sided random walk but the jumps leading to $x_2^2=x_1^1$ are forced to happen (we say that $x_2^2$ is pushed by $x_1^1$).
  \item Generally, $x_k^n$ performs a one-sided random walk (with jump probability~$p$), except that it is blocked by $x_k^{n-1}$ and is pushed by $x_{k-1}^{n-1}$ (whenever they exists). The updates are made from lower to upped levels.
\end{itemize}

\begin{figure}
\begin{center}
  \psfrag{L32}[l]{$\Lambda^{3}_{2}$}
  \psfrag{L21}[l]{$\Lambda^{2}_{1}$}
  \psfrag{P3}[c]{$P_3$}
  \psfrag{P2}[c]{$P_2$}
  \psfrag{P1}[c]{$P_1$}
  \psfrag{X3}[c]{$x^3(t)$}
  \psfrag{X2}[c]{$x^2(t)$}
  \psfrag{X1}[c]{$x^1(t)$}
  \psfrag{Y3}[c]{$x^3(t+1)$}
  \psfrag{Y2}[c]{$x^2(t+1)$}
  \psfrag{Y1}[c]{$x^1(t+1)$}
  \includegraphics[height=4.5cm]{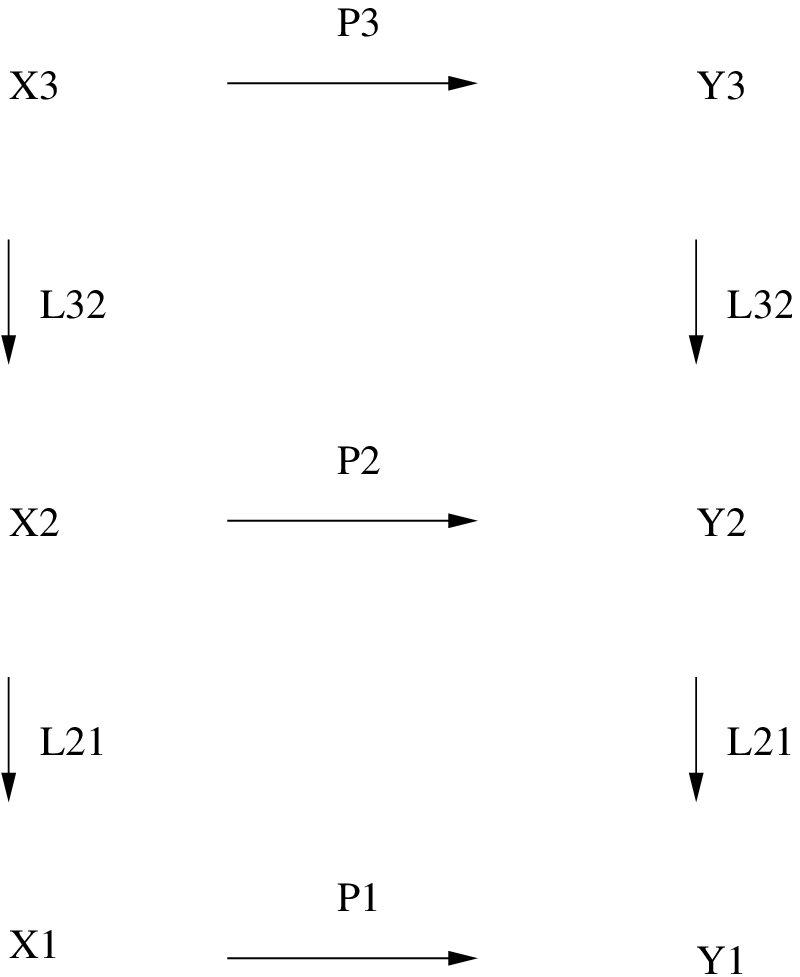}
\caption{Scheme for the sequential update dynamics.}
\label{FigSequential}
\end{center}
\end{figure}

\subsubsection*{Conserved measures}
Given the Markov chain on $\mathbb{GT}_N$ described above, it is of interest to know which class of measures are conserved by the time evolution.
\begin{prop}[See Proposition 2.5 of~\cite{BF08}]\label{PropConservationMeasureType}
Let $\mu_N(x^N)$ a probability measure on $W_N$. Consider the evolution of the measure
\begin{equation}\label{eq2.20}
M_N(X^N)=\mu_N(x^N)\Lambda^N_{N-1}(x^N,x^{N-1})\Lambda^{N-1}_{N-2}(x^{N-1},x^{N-2})\cdots\Lambda^2_1(x^2,x^1)
\end{equation}
on $\mathbb{GT}_N$ under the Markov chain $P^N_{\Lambda}$. Then the measure at time $t$ is given by
\begin{multline}
(M_N \underbrace{P^N_{\Lambda}\cdots P^N_{\Lambda}}_{t\textrm{ times}})(\tilde X^N)\\
=(\mu_N \underbrace{P_N\cdots P_N}_{t\textrm{ times}})(\tilde x^N)\Lambda^N_{N-1}(\tilde x^N,\tilde x^{N-1})\Lambda^{N-1}_{N-2}(\tilde x^{N-1},\tilde x^{N-2})\cdots\Lambda^2_1(\tilde x^2,\tilde x^1).
\end{multline}
\end{prop}
For this statement to hold we use the intertwining property (\ref{eqIntertwining}). The proof is elementary.

An example of measure of the form of Proposition~\ref{PropConservationMeasureType} is the \emph{packed initial conditions}
\begin{equation}\label{eqPackedIC}
x^n_k(t=0)=-n+k-1, \quad 1\leq k\leq n\leq N.
\end{equation}
This initial condition is, due to the interlacing enforced by the $\Lambda^n_{n-1}$'s, obtained by setting $\mu_N(x^N)=\delta_{x^N,(-N,\ldots,-1)}$. It is interesting to note that such a measure has a determinantal form.
\begin{lem}[See proof of Theorem 2.25 of~\cite{BF08}]\label{lemPackedIC}
Consider the probability measure on $W_N$ given by
\begin{equation}
\mu_N(x^N)=\const\, \Delta_N(x^N) \det(\Psi_i(x_j^N))_{i,j=1}^N
\end{equation}
where
\begin{equation}\label{eq3.35}
\Psi_i(x)=\frac{1}{2\pi\I}\oint_{\Gamma_0}\D z (1-z)^{N-i}z^{x+i-1},\quad 1\leq i\leq N.
\end{equation}
Then $\mu_N(x^N)=\delta_{x^N,(-N,\ldots,-1)}$.
\end{lem}
This statement is a key for a more difficult Theorem~\ref{ThmAztecMeasure} below.

\subsubsection{The parallel update dynamics}
Now we define a discrete time Markov chain that we call \emph{parallel update}, since the update occurs in parallel for each level and for the update of level $n$ we use the values of level $n-1$ at the previous time instant. In particular, the projection to the particles $\{x_1^1,x_2^2,\ldots\}$ is the parallel update dynamics of the totally asymmetric simple exclusion process (TASEP).

The state space of the Markov chain is
\begin{equation}
\mathbb{S}_N=\{x^1\in W_1,\ldots,x^n\in W_n\,|\, \Delta^k_{k-1}(x^{k-1},x^k)>0, 2\leq k \leq N\}.
\end{equation}
In our case, $\mathbb{S}_N$ is as $\mathbb{GT}_N$ except that now the particles are weakly interlacing, i.e., the strict inequalities in (\ref{eq3.2}) have to be replaced by inequalities. There is an analogue of Lemma~\ref{LemInterlacing} (that we are not using here), see~\cite{BFS07b}. Define the transition probabilities of a Markov chain on $\mathbb{S}_N$ by (we use the notation $X^N(t)=(x^1(t),\ldots,x^N(t))$)
\begin{multline}\label{eqMConGTnDelta}
P_{\Delta}^N(X^N(t),X^N(t+1))
=P_{1}(x^1(t),x^1(t+1))\\
\times\prod_{k=2}^N\frac{P_{k}(x^k(t),x^k(t+1))\Lambda^k_{k-1}(x^k(t+1),x^{k-1}(t))}{\Delta^k_{k-1}(x^k(t),x^{k-1}(t))},
\end{multline}
for $X^N(t),X^N(t+1)\in \mathbb{S}_N$, see Figure~\ref{FigParallel} for an illustration.

With our choice of $P_n$'s and $\Lambda^n_{n-1}$'s introduced above, the dynamics is similar to the one in sequential update: $x_k^n$ performs a one-sided random walk (with jump probability $p$), with the constraints that $x_k^{n}$ is blocked by $x_k^{n-1}$ and is pushed by $x_{k-1}^{n-1}$. However, this time, the update of the particles is parallel, i.e., first one checks which particles could jump because they are not blocked and then these free particles independently jump to their right with probability $p$, pushing, when required by interlacing, other particles at higher levels.
\begin{figure}
\begin{center}
  \psfrag{L32}[l]{$\Lambda^{3}_{2}$}
  \psfrag{L21}[l]{$\Lambda^{2}_{1}$}
  \psfrag{P3}[c]{$P_3$}
  \psfrag{P2}[c]{$P_2$}
  \psfrag{P1}[c]{$P_1$}
  \psfrag{X3}[c]{$x^3(t)$}
  \psfrag{X2}[c]{$x^2(t)$}
  \psfrag{X1}[c]{$x^1(t)$}
  \psfrag{Y3}[c]{$x^3(t+1)$}
  \psfrag{Y2}[c]{$x^2(t+1)$}
  \psfrag{Y1}[c]{$x^1(t+1)$}
  \includegraphics[height=4.5cm]{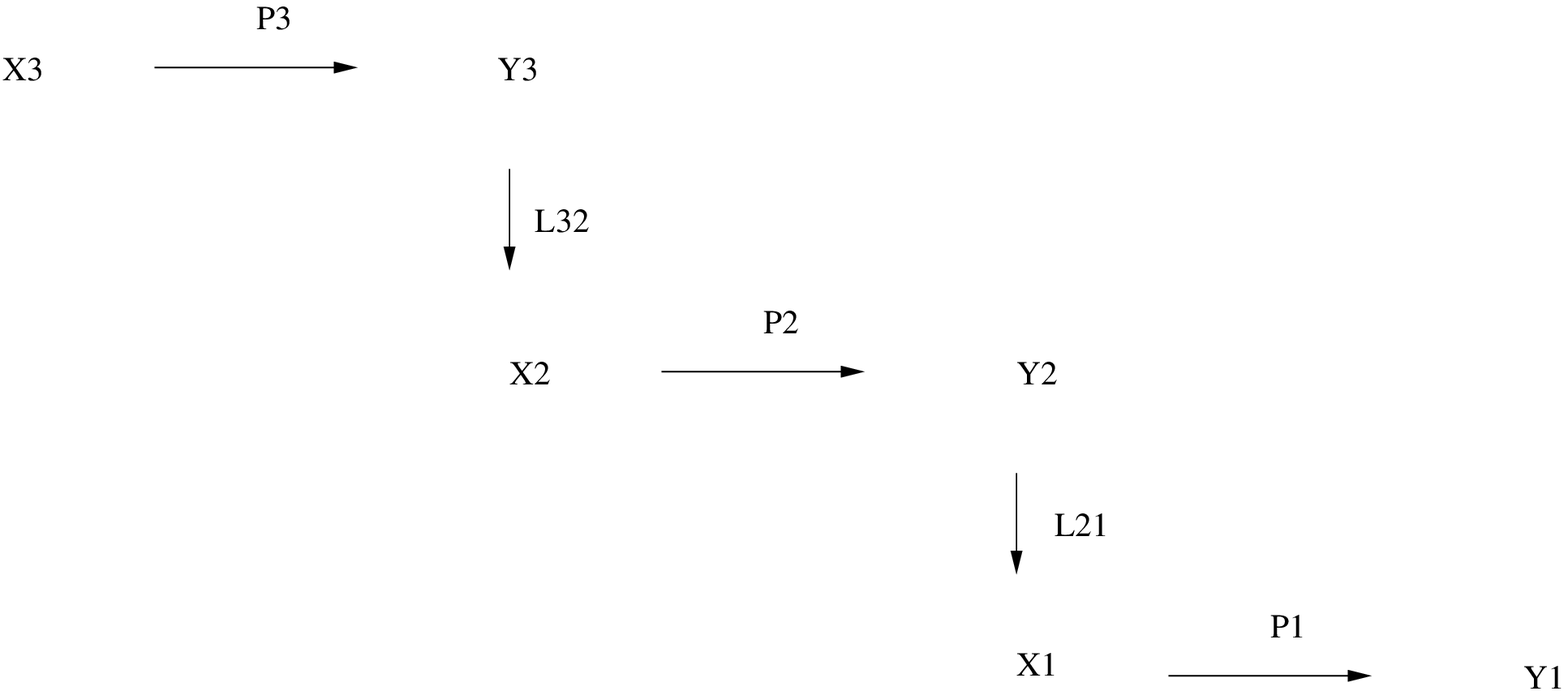}
\caption{Scheme for the parallel update dynamics.}
\label{FigParallel}
\end{center}
\end{figure}

The analogue of Proposition~\ref{PropConservationMeasureType} for the Markov chain on $\mathbb{S}_N$ with parallel update is the following.
\begin{prop}[See Proposition 2.5 of~\cite{BF08}]\label{PropConservationMeasureTypeParallel}
Let $\mu_N(x^N)$ a probability measure on $W_N$. Consider the evolution of the measure
\begin{equation}\label{eq2.20B}
M_N(X^N)=\mu_N(x^N)\Delta^N_{N-1}(x^N,x^{N-1})\Delta^{N-1}_{N-2}(x^{N-1},x^{N-2})\cdots\Delta^2_1(x^2,x^1)
\end{equation}
on $\mathbb{S}_N$ under the Markov chain $P^N_{\Delta}$. Then the measure at time $t$ is given by
\begin{multline}
(M_N \underbrace{P^N_{\Delta}\cdots P^N_{\Delta}}_{t\textrm{ times}})(\tilde X^N)\\
=(\mu_N \underbrace{P_N\cdots P_N}_{t\textrm{ times}})(\tilde x^N)\Delta^N_{N-1}(\tilde x^N,\tilde x^{N-1})\Delta^{N-1}_{N-2}(\tilde x^{N-1},\tilde x^{N-2})\cdots\Delta^2_1(\tilde x^2,\tilde x^1).
\end{multline}
\end{prop}

\subsubsection*{Parallel update - extended state space}
In order to obtain the connection with random tiling of the Aztec diamond, we need to consider two small extensions. First of all, the transition probabilities can be taken time-dependent as long as the intertwining property (\ref{eqIntertwining}) holds. We introduce the ''time''-variable $\tau$, which denotes the horizontal position in the scheme as in Figure~\ref{FigParallelExended}, where as a reference point, $\tau=0$, is set such that $P_1(0)$ is the transition of $x^1(0)$ to $x^1(1)$.
\begin{figure}
\begin{center}
  \psfrag{L32a}[l]{$\Lambda^{3}_{2}$}
  \psfrag{L21a}[l]{$\Lambda^{2}_{1}$}
  \psfrag{L32}[l]{$\Lambda^{3}_{2}$}
  \psfrag{L21}[l]{$\Lambda^{2}_{1}$}
  \psfrag{P2a}[c]{$P_2(t-2)$}
  \psfrag{P1a}[c]{$P_1(t-1)$}
  \psfrag{P3}[c]{$P_3(t-2)$}
  \psfrag{P2}[c]{$P_2(t-1)$}
  \psfrag{P1}[c]{$P_1(t)$}
  \psfrag{X3}[c]{$x^3(t)$}
  \psfrag{X2a}[c]{$y^3(t)$}
  \psfrag{X1a}[c]{$y^2(t)$}
  \psfrag{X2}[c]{$\begin{array}{c}x^2(t)\\=y^3(t+1)\end{array}$}
  \psfrag{X1}[c]{$\begin{array}{c}x^1(t)\\=y^2(t+1)\end{array}$}
  \psfrag{Y3}[c]{$x^3(t+1)$}
  \psfrag{Y2}[c]{$x^2(t+1)$}
  \psfrag{Y1}[c]{$x^1(t+1)$}
  \includegraphics[height=4.5cm]{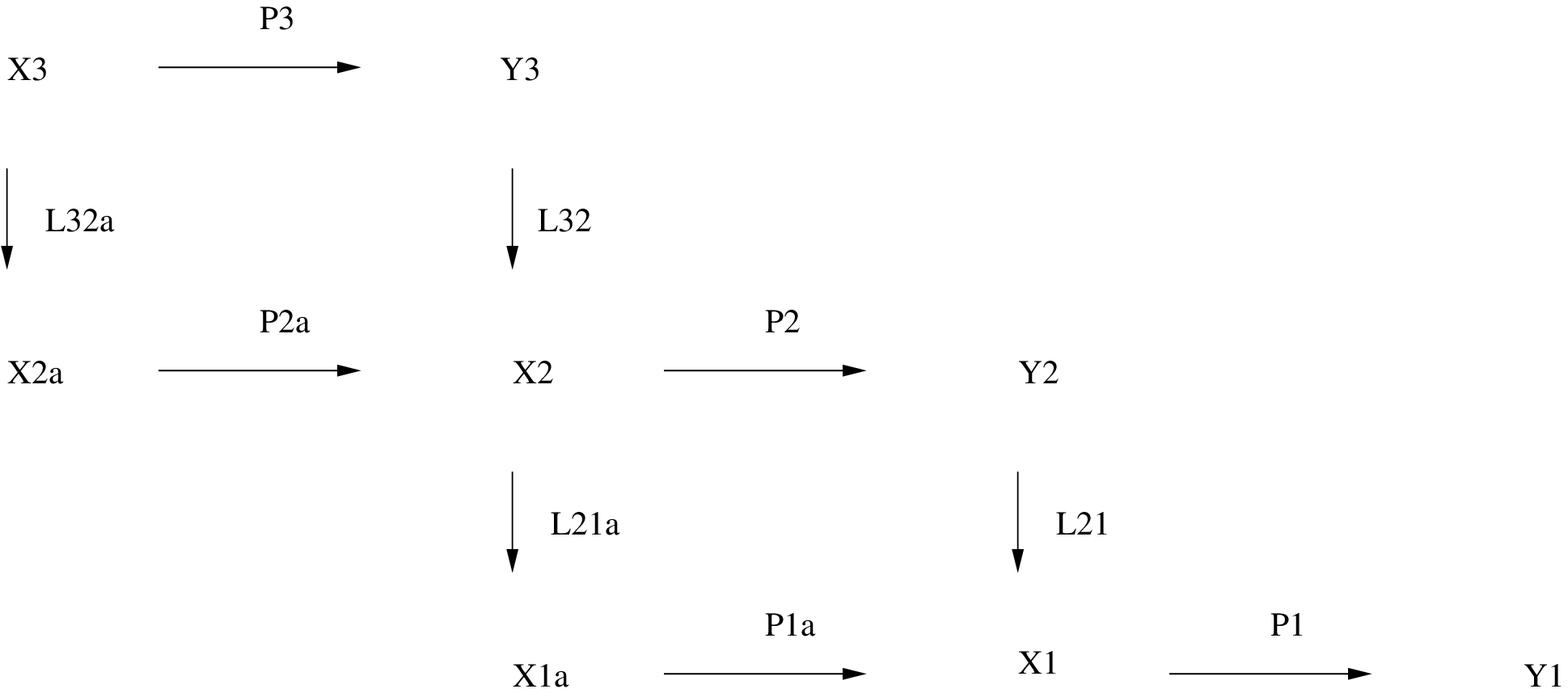}
\caption{Scheme for the parallel update dynamics with extended space time and time-dependent transition probabilities.
The intertwining condition reads, $\Lambda^n_{n-1} P_{n-1}(\tau) = P_n(\tau) \Lambda^n_{n-1}$.}
\label{FigParallelExended}
\end{center}
\end{figure}
Secondly, we extend the state space. Instead of considering the measure (\ref{eq2.20B}) as the one induced by the measure $\mu_N$ on $W_N$ through the action of consecutive actions of $\Delta^n_{n-1}$, we consider it as the induced measure by the consecutive actions of $\Lambda^n_{n-1}$ and $P_{n-1}$. The middle points of the chain $(\Lambda^n_{n-1}\circ P_{n-1})(x^n,x^{n-1})$ is denoted by $y^n$. The choice of using the index $n$ rather than $n-1$ is slightly arbitrary, but the reason lies in the fact that it is natural to use the same index corresponding to the same ``time''-$\tau$ in the generalizations (see next section too). Thus we consider the extended state space
\begin{equation}
\mathbb{S}^{\rm ext}_N=\{X^N\in \mathbb{S}_N, Y^N\in \mathbb{S}_{N-1}\,|\, x^k,y^k\in W_k, y^k \prec x^k, 2\leq k \leq N\}.
\end{equation}
where $X^N=(x^1,\ldots,x^N)$ and $Y^N=(y^1,\ldots,y^{N-1})$. The induced one-time transition of this extended state space can be described as following:
\begin{itemize}
\item the $y$-particles are deterministically updated: $y^{k+1}(t+1)=x^k(t)$,
\item the $x$-particles are updated according to the parallel update (\ref{eqMConGTnDelta}),
\end{itemize}
see Figure~\ref{FigParallelExended}. Also in this setting there is a measure that is conserved.
\begin{prop}\label{PropConservationMeasureTypeParallelExtended}
Let $\mu_N(x^N)$ a probability measure on $W_N$. Consider the evolution of the measure at time $t=0$ given by
\begin{multline}
\mu_N(x^N)\Lambda^N_{N-1}(x^N,y^{N})P_{N-1}(\tau)(y^{N},x^{N-1}) \\
\times \Lambda^{N-1}_{N-2}(x^{N-1},y^{N-1})P_{N-2}(\tau+1)(y^{N-1},x^{N-2})\cdots\\
\times\Lambda^2_1(x^2,y^2)P_1(\tau+N-2)(y^2,x^1)
\end{multline}
under the Markov chain with parallel update in the extended state space. Then, the measure at time $t$ is given by
\begin{multline}
(\mu_N P_N(\tau)\cdots P_N(\tau+t-1))(\tilde x^N)\Lambda^N_{N-1}(\tilde x^N,\tilde y^{N})P_{N-1}(\tau+t)(\tilde y^{N},\tilde x^{N-1}) \\
\times \Lambda^{N-1}_{N-2}(\tilde x^{N-1},\tilde y^{N-1})P_{N-2}(\tau+t+1)(\tilde y^{N-1},\tilde x^{N-2})\cdots\\
\times\Lambda^2_1(\tilde x^2,\tilde y^2)P_1(\tau+t+N-2)(\tilde y^2,\tilde x^1).
\end{multline}
\end{prop}

The proof of Proposition~\ref{PropConservationMeasureTypeParallelExtended} is almost identical to the one of Propositions~\ref{PropConservationMeasureType} and~\ref{PropConservationMeasureTypeParallel} (see Proposition 2.5 of~\cite{BF08}).

\subsection{Line ensembles associated with parallel update}
The map from a configuration in the extended state space to a set of non-intersecting line ensembles is simple. The LGV graph is the one of Figure~\ref{FigAztec4}. First let us describe where the particles are located and then we will explain how to add the lines joining them. The $x$-particles are at positions
$$\{(2k-1,x^k_j),1\leq j \leq k\leq N\}$$
and the $y$-particles are at positions
$$\{(2k,y^{k+1}_j),1\leq j \leq k\leq N-1\}.$$
At this point we can join the top particles at every time with a line following the LGV graph, then join the top remaining particles and so on. This results in a set of non-intersecting line ensembles as the solid lines in Figure~\ref{FigAztec4} (right) (the deterministically added dashed lines are not yet present, but can be now added).

The above discussion concerns the dynamics only and therefore it is clear that to obtain the Aztec diamond distributed according to the measure (\ref{eqMeasAztec}) one has to start with a special initial condition and also to choose the transition matrices $P_n(\tau)$ in a $\tau$-dependent way such that particles at level $n$ can move for the first time at time moment $n$. The initial condition is the packed initial condition described above, see Lemma~\ref{lemPackedIC}, and the coefficients in the transition matrices $P_n(\tau)$ are chosen to be $p=a^2/(1+a^2)$ for $\tau\geq 0$ and $p=0$ for $\tau<0$, i.e., with $\bf P$ in Definition~\ref{DefDiscreteCharlier} given by
\begin{equation}\label{eq3.31}
\begin{aligned}
{\bf P}(\tau)(x,y)&=\frac{a^2}{1+a^2}\delta_{x+1,y}+\frac{1}{1+a^2}\delta_{x,y},\quad \tau\geq 0,\\
{\bf P}(\tau)(x,y)&=\delta_{x,y},\quad \tau<0.
\end{aligned}
\end{equation}
The dynamics on the line ensembles is then exactly the one described at the end of Section~\ref{SectAztec}. Indeed, first of all the initial conditions are the same. Then, given the values after $n$ steps, the parallel dynamics first update deterministically the values of the $y$-particles and then of the $x$-particles. These latter ores try to jump to their right with probability $p$, which is non-zero only for levels below $n+1$ (compare Figure~\ref{FigParallelExended}). The jumps are allowed only if the arrival configuration satisfies interlacing, i.e., non-intersecting in terms of line ensembles.

\begin{thm}\label{ThmAztecMeasure}
The probability measure on the line ensembles (i.e., on the extended state space) corresponding to the Aztec diamond of size $N$ is given by
\begin{equation}
\begin{aligned}
{\rm const}\times \det[\tilde \Psi_j^N(x_i^N)]_{i,j=1}^N \det[\phi(x_i^N,y_j^{N})]_{i,j=1}^N \det[f(y^{N}_i-x^{N-1}_j)]_{i,j=1}^{N-1}\\
\times \det[\phi(x_i^{N-1},y_j^{N-1})]_{i,j=1}^{N-1} \det[f(y^{N-1}_i-x^{N-2}_j)]_{i,j=1}^{N-2} \\
\cdots\times \det[\phi(x_i^2,y_j^{2})]_{i,j=1}^2 \det[f(y^{2}_i-x^{1}_j)]_{i,j=1}^{1},
\end{aligned}
\end{equation}
where $\tilde \Psi_j^N(x)=\frac{1}{2\pi\I}\oint_{\Gamma_0} \D z (1-z)^{N-i} z^{x+i-1}(1-p+p z^{-1})$ with $p=a^2/(1+a^2)$, and $f(x)=p \delta_{x,1}+(1-p)\delta_{x,0}$, and $\phi$ given in Lemma~\ref{LemInterlacing}.
\end{thm}
This result is simply obtained by plugging in Theorem~\ref{PropConservationMeasureTypeParallelExtended} the explicit expressions of the packed initial condition (see Lemma~\ref{lemPackedIC}), the determinantal form of the $\Lambda^n_{n-1}$ of Lemma~\ref{LemInterlacing} and of $P_N$ (see Definition~\ref{DefDiscreteCharlier} with (\ref{eq2.11}) and (\ref{eq3.31})).

\begin{rem}
It is well known that a measure of the type of Theorem~\ref{SectAztec} is determinantal as it fits in the general framework of \emph{conditional $L$-ensembles} developed in~\cite{RB04} (see for example Theorem 4.1 of~\cite{BF07} for a continuous case analogue). From this one can start computing a lot of observables and, in particular, obtain the law of large numbers / central limit theorem type results mentioned in the introduction.
\end{rem}

\section{Generalizations}\label{SectGeneralization}
In this section we present a generalization of the Aztec diamond and another random tiling model. Both have non-intersecting lines representations with Markovian dynamics induced by a shuffling algorithm, which can be linked with the intertwining Markov chains discussed in Section~\ref{SectAKPZ}.

There are other random tiling models with non-intersecting line representations. One example is the double Aztec diamond studied in~\cite{AJvM14,ACJM13}. They might fit in the general framework, but it is not so straightforward as the underlying LGV graph is not translation-invariant. In that case, the transition probabilities are not translation-invariant and thus Proposition~\ref{propApp3} can not be used to verify the intertwining property.

\subsection{Aztec diamond with non-uniform weights}
In the original paper of Elkies, Kuperbert, Larsen and Propp~\cite{EKLP92}, they considered more general weights on random tilings of the Aztec diamond, namely
\begin{equation}\label{eqAztecVol}
\widetilde \Prob_N(T)\propto a^{v(T)} q^{r(T)}
\end{equation}
where $v(T)$ is the number of vertical tiles and $r(T)$ is what they called \emph{rank}. The rank is the minimal number of elementary moves (switch of a pair of vertical dominoes into a pair of horizontal dominoes, or vice versa) that is needed to generate the tiling $T$ from the tiling $T_0$ consisting of exclusively horizontal dominoes. The value of $r(T)$ can be read off the line ensembles. For the Aztec diamond of size $N$, only the top $N$ lines, call them $\ell_1(T),\ldots,\ell_N(T)$, are not deterministically set to be straight lines (see for instance Figure~\ref{FigAztec4}). The line ensemble of the tiling $T_0$ consists of all straight lines. Then, $r(T)=\sum_{n=1}^N {\rm Area}(\ell_n(T)-\ell_n(T_0))$. This model was further studied in~\cite{CY14,BBCCR15}.

From the line ensembles of the Aztec diamond of size $N$, it is easy to verify that the weight $\widetilde \Prob_N$ can be obtained from the LGV graph made by $N$ copies of the basic graph but with non-homogeneous weights. More precisely, the resulting LGV graph is represented in Figure~\ref{FigLGVAztecVolume}.
\begin{figure}
\begin{center}
\psfrag{a1}[c]{$\alpha_1$}
\psfrag{a2}[c]{$\alpha_2$}
\psfrag{a3}[c]{$\alpha_3$}
\psfrag{an}[c]{$\alpha_N$}
\psfrag{b1}[c]{$\beta_1$}
\psfrag{b2}[c]{$\beta_2$}
\psfrag{b3}[c]{$\beta_3$}
\psfrag{bn}[c]{$\beta_N$}
\psfrag{y1}[c]{$\Omega_0$}
\psfrag{x1}[c]{$x^1$}
\psfrag{y2}[c]{$y^2$}
\psfrag{x2}[c]{$x^2$}
\psfrag{y3}[c]{$y^3$}
\psfrag{x3}[c]{$x^3$}
\psfrag{y4}[c]{$y^4$}
\psfrag{yN}[c]{$y^N$}
\psfrag{xN}[c]{$x^N$}
\psfrag{yNp1}[c]{$\Omega_0$}
\includegraphics[height=5cm]{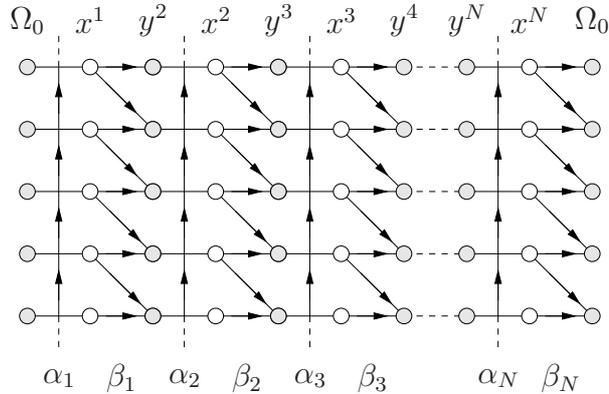}
\caption{LGV graph for the Aztec diamond distributed according to (\ref{eqAztecVol}). The vertical edges have weights $\alpha_1=a^2 q^{2N}, \alpha_2=a^2 q^{2N-2},\ldots,\alpha_N=a^2 q^2$, the diagonal edges have weights $\beta_1=1/q^{2N-1},\beta_2=1/q^{2N-3},\ldots,\beta_N=1/q$, and all other edges have weight $1$.}
\label{FigLGVAztecVolume}
\end{center}
\end{figure}

This line ensemble is the one generated by the Markov dynamics with parallel update, with packed initial condition, and the transition matrices given as follows. Consider parameters $\alpha_1,\alpha_2,\ldots$ all different (the case where two or more are equal is obtained as a limit): for $x^n,y^{n+1}\in W_n$,
\begin{equation}\label{eqPnGen}
P_n(\tau)(x^n,y^{n+1})=\frac{\det[\alpha_i^{y_j^{n+1}}]_{i,j=1}^n}{\det[\alpha_i^{x_j^n}]_{i,j=1}^n}
\frac{\det[f_n(x_i^n-y_j^{n+1},\tau)]_{i,j=1}^n}{\prod_{j=1}^n(1+\beta_n(\tau)/\alpha_j)},
\end{equation}
with
\begin{equation}
f_n(m,\tau)= \delta_{m,0}+\beta_n(\tau)\delta_{m,-1},
\end{equation}
where
\begin{equation}
\beta_n(\tau)=\left\{
\begin{array}{cc}\beta_n&\textrm{ for }\tau\geq 0,\\
0,&\textrm{ for }\tau<0.
\end{array}\right.
\end{equation}
One can easily verify that, for $\tau\geq 0$,
\begin{equation}
\det[f_n(x_i^n-y_j^{n+1},\tau)]_{i,j=1}^n =
\left\{\begin{array}{ll}
\prod_{j=1}^n \beta_n^{y_j^{n+1}-x_j^n},&\textrm{ if }y_j^{n+1}-x_j^n\in\{0,1\}, 1\leq j \leq n,\\
0&\textrm{ otherwise},
\end{array}\right.
\end{equation}
while for $\tau<0$ is it just $\det[f_n(x_i^n-y_j^{n+1},\tau)]_{i,j=1}^n=\prod_{j=1}^n \delta_{x_j^n,y_j^{n+1}}$.

Further, the intertwining is satisfied for the following choice of $\Lambda^n_{n-1}$ (see Proposition~\ref{propApp3}): for $x^n\in W_n$ and $y^{n}\in W_{n-1}$, define
\begin{equation}\label{eqLnGen}
\Lambda^n_{n-1}(x^n,y^n) =\frac{\det[\alpha_i^{y_j^n}]_{i,j=1}^{n-1}}{\det[\alpha_i^{x_j^n}]_{i,j=1}^n}\frac{\det[\tilde f_n(x_i^n-y_j^n)]_{i,j=1}^n}{\prod_{j=1}^{n-1}(1-\alpha_n/\alpha_j)^{-1}},
\end{equation}
where $y_n^n={\rm virt}$ is a virtual variable in order to write the formula in a more compact way, $\tilde f_n(x-{\rm virt})=\alpha_n^x$, and
\begin{equation}
\tilde f_n(m)=\sum_{k\geq 0}\alpha_n^m \delta_{m,k}.
\end{equation}

\subsection{Random tiling of a Tower}
In this section we present a generalization that fits in our general Markov chain dynamics, then we will explain the induced dynamics on a set of non-intersecting line ensembles, and finally we will explain how this generalizes to a shuffling-type algorithm for an associated random tiling model.

\subsubsection*{Markov chains}
Now we consider a Markov chain on the scheme of Figure~\ref{FigGeneral}. It differs from the one of the Aztec diamond represented in Figure~\ref{FigParallelExended} by the fact that the transition between $x^n$ and $x^{n-1}$ is a sequence of two $\Lambda$ and one $P$ transition. The middle points of the chain between $x^n$ and $x^{n-1}$ are denoted by $y^n$ and $z^n$ respectively.
\begin{figure}
\begin{center}
  \psfrag{X20}[c]{$x^2(0)$}
  \psfrag{X20b}[c]{$x^2(1)$}
  \psfrag{Y20}[c]{$y^2(0)$}
  \psfrag{Y20b}[c]{$y^2(1)$}
  \psfrag{Z20}[c]{$z^2(0)$}
  \psfrag{Z21}[c]{$x^1(0)=z^2(1)$}
  \psfrag{X21b}[c]{$x^1(1)$}
  \psfrag{Y10}[c]{$y^1(0)$}
  \psfrag{Y10b}[c]{$y^1(1)$}
  \psfrag{L43}[l]{$\Lambda^4_3$}
  \psfrag{L32}[l]{$\Lambda^3_2$}
  \psfrag{L43b}[l]{$\Lambda^4_3$}
  \psfrag{L32b}[l]{$\Lambda^3_2$}
  \psfrag{L21}[l]{$\Lambda^2_1$}
  \psfrag{P4}[c]{$P_4(-1)$}
  \psfrag{P4b}[c]{$P_4(0)$}
  \psfrag{P3}[c]{$P_3(-1)$}
  \psfrag{P3b}[c]{$P_3(0)$}
  \psfrag{P2}[c]{$P_2(-1)$}
  \psfrag{P2b}[c]{$P_2(0)$}
  \psfrag{P1}[c]{$P_1(0)$}
  \includegraphics[height=5.5cm]{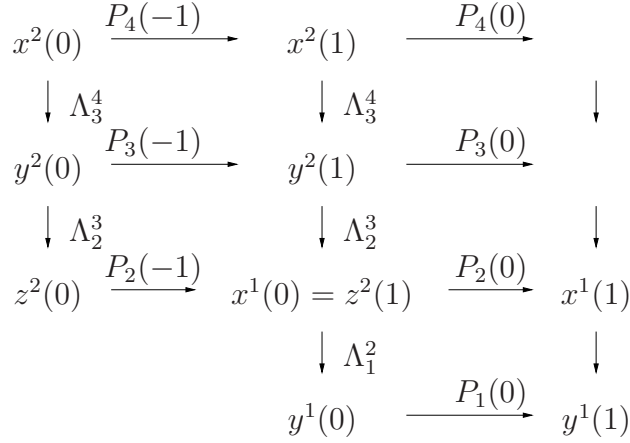}
\caption{Scheme of the Markov chain for the random tiling of a Tower.}
\label{FigGeneral}
\end{center}
\end{figure}
Specifically, let us consider the case where the transition are:
\begin{itemize}
\item from $x^n(t)$ to $y^n(t)$: $\Lambda^{2n}_{2n-1}$,
\item from $y^n(t)$ to $z^n(t)$: $\Lambda^{2n-1}_{2n-2}$,
\item from $z^n(t)$ to $x^{n-1}(t)$: $P_{2n-2}(t-n+1)$,
\end{itemize}
where $P_n(\tau)$ and $\Lambda^n_{n-1}$ is given as in (\ref{eqPnGen}) and (\ref{eqLnGen}) in the limit $\alpha_1=\alpha_3=\ldots=\alpha$, $\alpha_2=\alpha_4=\ldots=\tilde \alpha$, and $\beta_1=\beta_2=\ldots=\beta$.

We consider packed initial conditions for the $x^1,x^2,\ldots$, which by interlacing implies also $z^n(0)=y^n(0)=x^n(0)$, and define the evolution of the Markov chain is as follows:
\begin{itemize}
\item update in parallel $z^{n+1}(t+1)=x^{n}(t)$,
\item update $y^{n+1}(t+1)$ as the middle point of the two-step chain \mbox{$(P_{2n-1}(t-n+1)\circ \Lambda^{2n-1}_{2n-2})(y^{n+1}(t),x^n(t))$},
\item update $x^{n+1}(t+1)$ as the middle point of the two-step chain \mbox{$(P_{2n}(t-n+1)\circ \Lambda^{2n}_{2n-1})(x^{n+1}(t),y^n(t+1))$}.
\end{itemize}

\begin{thm}\label{ThmTowerMeasure}
The probability measure on the line ensembles (i.e., on the extended state space) corresponding to the Tower of size $N$ is given by
\begin{multline}
{\rm const}\times \det[\tilde\Psi^{2N}_i(x^N_j)]_{i,j=1}^{2N} \det[\phi(x^N_i,y^N_j)]_{i,j=1}^{2N} \det[\phi(y^N_i,z^N_j)]_{i,j=1}^{2N-1}\\
\times \det[\tilde f(x^{N-1}_i-z^N_j)]_{i,j=1}^{2N-2}  \det[\phi(x^{N-1}_i,y^{N-1}_j)]_{i,j=1}^{2N-2} \det[\phi(y^{N-1}_i,z^{N-1}_j)]_{i,j=1}^{2N-3}\\
\times\cdots\times \det[\tilde f(x^{1}_i-z^2_j)]_{i,j=1}^{2} \det[\phi(x^{1}_i,y^{1}_j)]_{i,j=1}^{2} \prod_{n=1}^N F(x^n,y^n,z^n),
\end{multline}
where $\tilde\Psi^{2N}_i(x)=\frac1{2\pi\I}\oint_{\Gamma_0}\D z (1-z)^{2N-i}z^{x+i-1}(1+\beta/z)$, $\tilde f(x)=\delta_{x,0}+\beta \delta_{x,1}$, and $\phi$ is given in Lemma~\ref{LemInterlacing} with $y^n_{2n}={\rm virt}$ as well as $z^n_{2n-1}={\rm virt}$. Further,
\begin{equation}
F(x^n,y^n,z^n)=\prod_{i=1}^{2n} \alpha^{x^n_i}\prod_{j=1}^{2n-1} (\tilde\alpha/\alpha)^{y^n_j}\prod_{k=1}^{2n-2} (1/\tilde\alpha)^{z_i^n}.
\end{equation}
\end{thm}

\subsubsection*{Non-intersecting line ensembles}
\begin{minipage}{0.75\textwidth}
\hspace{1em} In a similar way as for the Aztec diamond, we can produce a set of non-intersecting line ensembles as follows.
The basic Lindstr\"om-Gessel-Viennot (LGV) directed graph of this generalization is the following, where the first column of vertical edges have weights $\alpha$, the second column have weights $\tilde \alpha$, and the diagonal edges have weights $\beta$. Consider the graph obtained by $N$ copies of this basic LGV graph as in Figure~\ref{FigLGVGeneralN}.
\end{minipage}
\hfill
\begin{minipage}{0.22\textwidth}
\begin{center}
\includegraphics[height=4cm]{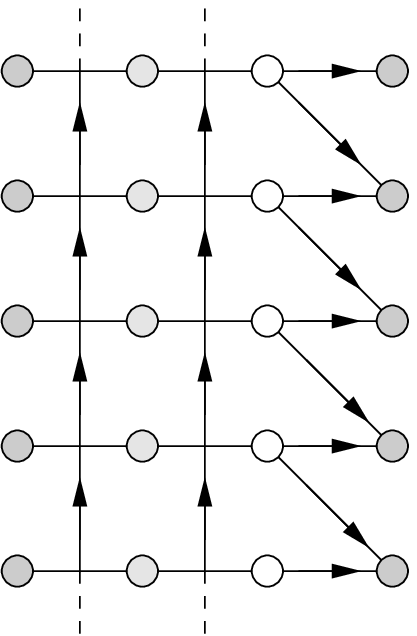}
\end{center}\vfill
\end{minipage}

\begin{figure}
\begin{center}
\psfrag{a1}[c]{$\alpha$}
\psfrag{a2}[c]{$\tilde \alpha$}
\psfrag{b}[c]{$\beta$}
\psfrag{z1}[c]{$\Omega_0$}
\psfrag{y1}[c]{$y^1$}
\psfrag{x1}[c]{$x^1$}
\psfrag{z2}[c]{$z^2$}
\psfrag{y2}[c]{$y^2$}
\psfrag{x2}[c]{$x^2$}
\psfrag{z3}[c]{$z^3$}
\psfrag{zN}[c]{$z^N$}
\psfrag{yN}[c]{$y^N$}
\psfrag{xN}[c]{$x^N$}
\psfrag{zNp1}[c]{$\Omega_0$}
\includegraphics[height=5cm]{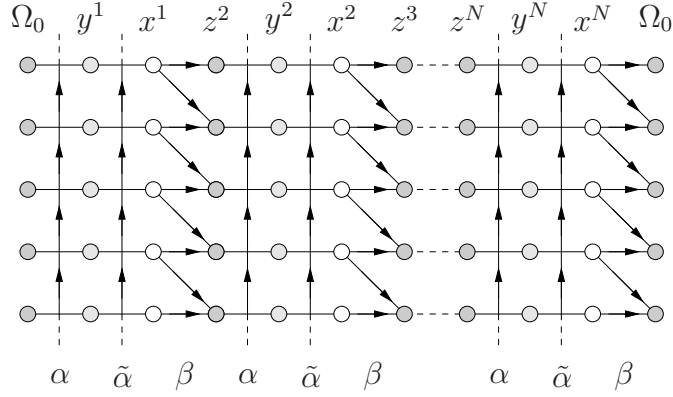}
\caption{LGV graph for the $\alpha-\alpha-\beta$ generalization. The vertical edges have weights $\alpha$, $\tilde \alpha$ and the diagonal edges have weight $\beta$. All other edges have weight $1$. The lines start and end from positions $\Omega_0=\{-1,-2,\ldots\}$.}
\label{FigLGVGeneralN}
\end{center}
\end{figure}

For a configuration $\{y^1,x^1,z^2,y^2,x^2,\ldots,z^N,y^N,x^N\}$, place the \mbox{$x$-particles} at positions
$$\{(3k-1,x^k_j),1 \leq k\leq N, 1\leq j \leq 2k\},$$
the $y$-particles at positions
$$\{(3k-2,y^k_j),1\leq k\leq N, 1\leq j\leq 2k-1\},$$
and the $z$-particles at positions
$$\{(3k-3,z^k_j),2\leq k\leq N, 1\leq j \leq 2k-2\}.$$
The ensemble of non-intersecting lines is the one where lines start (at line-time $0$) and end (at line-time $3N$) from $\{-1,-2,\ldots\}$. The weight of such a configuration generated by the dynamics described above is the following: vertical edges just before $y$-particles have weight $\alpha$, the ones just after the $y$-particles have weight $\tilde\alpha$, and the diagonal edges have weight $\beta$. All other weights are set to be $1$, see Figure~\ref{FigLinesGeneral} for an illustration.

\begin{figure}
\begin{center}
\includegraphics[height=6cm]{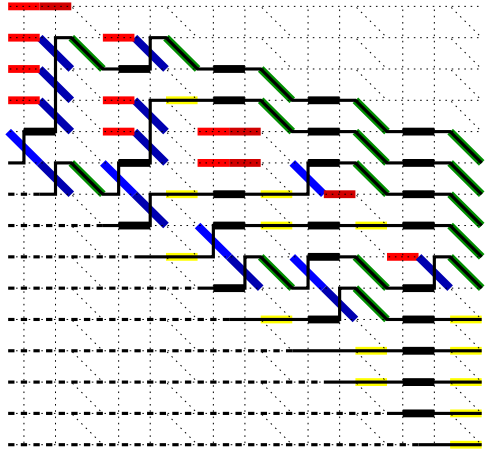}
\caption{Line ensembles decorated with domino colors as in Figure~\ref{FigRocketSmall} below.}
\label{FigLinesGeneral}
\end{center}
\end{figure}

\subsubsection*{Random tiling model}
The set of non-intersecting line ensembles described above and generated by the Markov chain, can be associated to a random tiling model. For this generalization it is easier to think of a random tiling as a perfect matching of a bipartite graph, instead of a covering of a region. For the Aztec diamond, the $N=1$ bipartite graph is illustrated in Figure~\ref{FigBasicRocket} (right). In the present case, the $N=1$ basic building bloc is the one in Figure~\ref{FigBasicRocket} (left) consisting of a hexagon and two squares. As it looks like a small tower, we call this random tiling of a Tower. \begin{figure}
\begin{center}
\includegraphics[height=4cm]{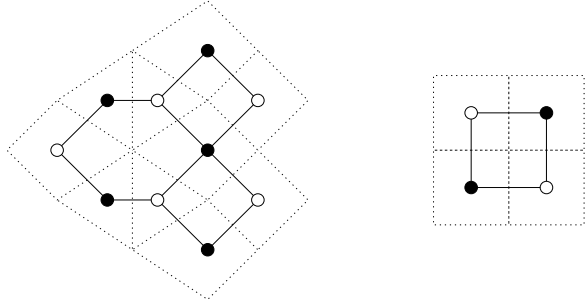}
\caption{The basic building bloc of the Tower (left) vs. the one of the Aztec diamond (right). One can think of random tiling of the plane with ``dominoes'' defined as two neighboring basic elements (squares, triangles, or kite for the Tower; only squares for Aztec) or, equivalently, think of \emph{perfect matching} of the dual graph inside the Tower / Aztec.}
\label{FigBasicRocket}
\end{center}
\end{figure}

The random tiling model of size $N$ consists in a perfect matching of the following graph: in the left-most column we have $N$ basic blocs where the hexagons share a horizontal edge, this is followed by $N+1$ basic blocs obtained by putting the left vertex of the hexagon at the right edge of the squares. This is repeated $N$ times, so that in the end we have $(3N-1)N/2$ hexagons. In Figure~\ref{FigRocketSmall} we show an example for $N=5$.
\begin{figure}
\begin{center}
\includegraphics[height=5cm]{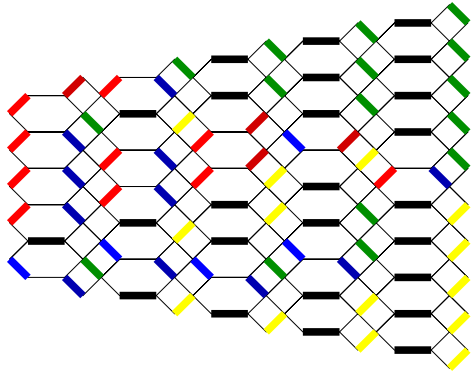}
\caption{An example of random tiling of a Tower of size $N=5$. There are two types of red/blue dominoes, depending on which $\alpha$ the corresponding column belongs to.}
\label{FigRocketSmall}
\end{center}
\end{figure}
The weight of a random tiling configuration is the following: each green domino has weight $\beta$, each blue domino has either weight $\alpha$ or $\tilde\alpha$ depending its horizontal coordinate.

\subsubsection*{Shuffling-type algorithm}
It is not too surprising that one can define a shuffling algorithm to generate the random tiling. For simplicity, consider the uniform tiling model. Start with a domino tiling of a Tower of size $N$.
\begin{itemize}
\item[Step 1:] Blue dominoes stay put. Red dominoes move up by one unit. Yellow dominoes move to the right and down by one unit. Green dominoes move right and up by one unit. Also, delete the dominoes whose trajectories intersect while doing these moves.
\item[Step 2:] What remains is a partial tiling of the Tower of size $N+1$. The remainder can be uniquely decomposed into disjoint pieces, each of one is one of the following $5$ configurations: (a) the basic building bloc, (b) hexagon plus a square (above or below), (c) hexagon, and (d) square. Then randomly tile all these pieces independently.
\end{itemize}
The result is a tiling of a Tower of size $N+1$ that had the desired distribution. In Figure~\ref{FigRocketLarge} we show a random tiling with uniform distribution of a Tower of size $N=100$, corresponding to $\alpha=\tilde\alpha=\beta=1$.
\begin{figure}
\begin{center}
\includegraphics[height=14cm,angle=-90]{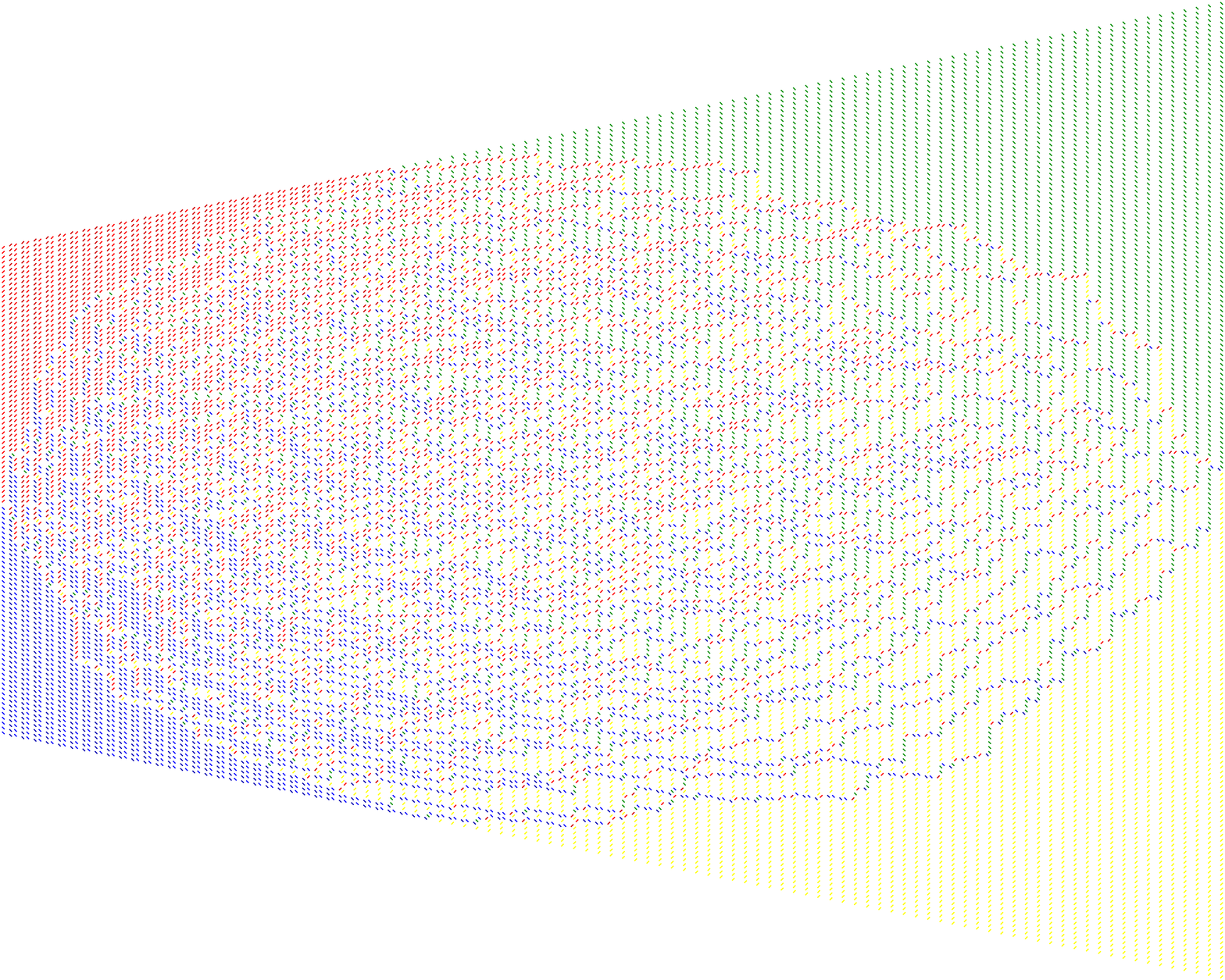}
\caption{An example of random tiling of a Tower of size $N=100$ with parameters $\alpha=\tilde \alpha=\beta=1$. For better visibility the background graph and the black dominoes are not shown. Also, the picture is rotated by 90 degrees with respect to the one in Figure~\ref{FigRocketSmall}.}
\label{FigRocketLarge}
\end{center}
\end{figure}

By choosing the parameters $\alpha$, $\tilde\alpha$, and $\beta$ not all equal to one, one can obtain different figures. For instance, one can get a heart-like limit shape as in Figure~\ref{FigRocketLargeBis}.
\begin{figure}
\begin{center}
\includegraphics[height=14cm,angle=-90]{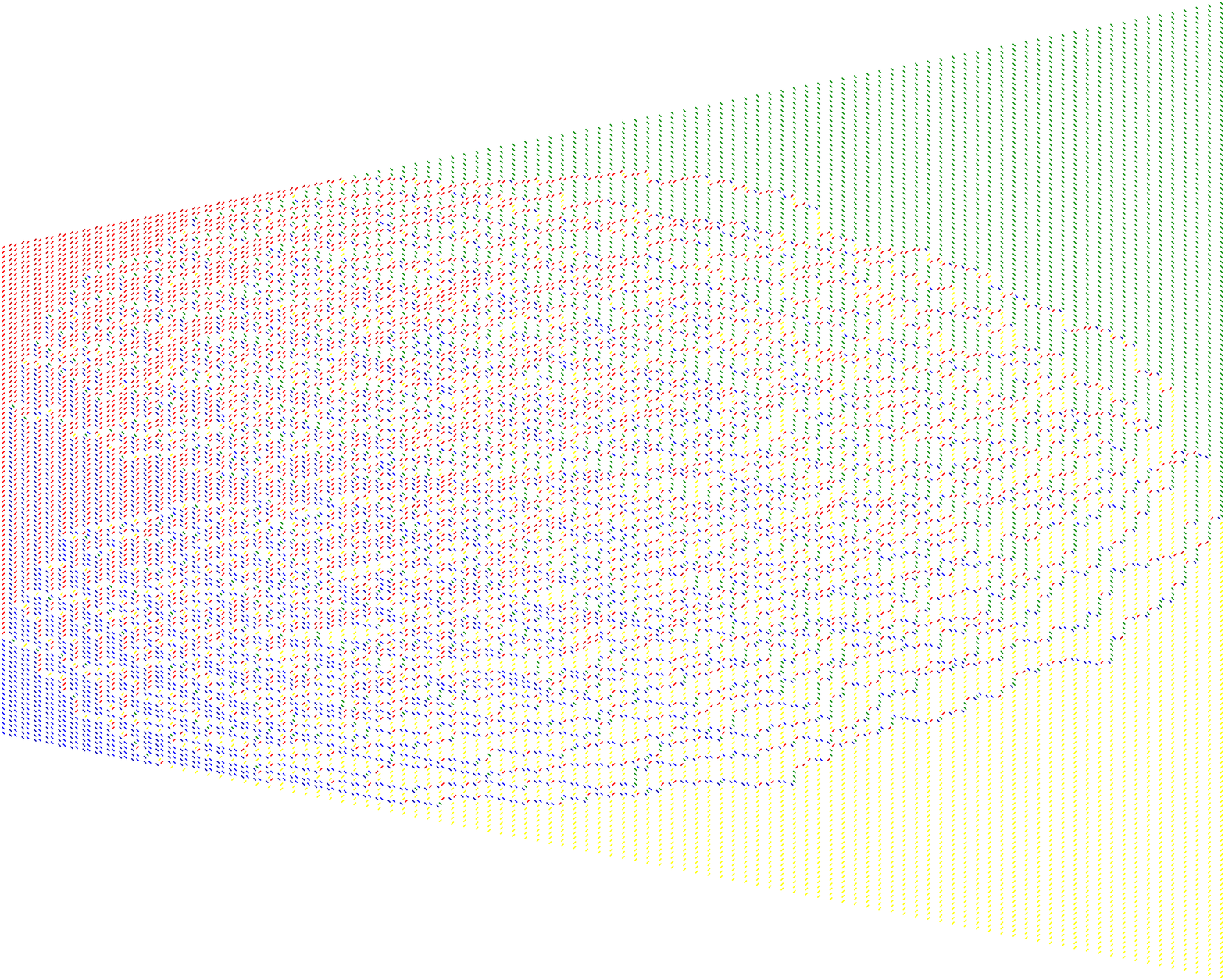}
\caption{An example of random tiling of a Tower of size $N=100$ with parameters $\alpha=1/4$, $\tilde \alpha=3/4$, $\beta=1$. For better visibility the background graph and the black dominoes are not shown.}
\label{FigRocketLargeBis}
\end{center}
\end{figure}
One can generalize the above setting to have $m>2$ instead of $2$ $\alpha$-blocs. In that case, by choosing different $\alpha$-values for the different weights, one can get for instance a limit shape with $m-1$ fjords instead of one (see~\cite{Pet12,DM14} for similar behaviors in lozenge tiling). In the global scaling, one expects to see Gaussian Free Field fluctuations~\cite{Ken01}, at the edges of the regular pieces one expects to see an Airy$_2$ process~\cite{PS02,Jo03}, and tip of the fjord one expects to see the Pearcey process~\cite{TW06}.

\newpage
\appendix

\section{Toeplitz-like transition probabilities}\label{AppToeplitz}
Here we give some results on transition probabilities which are translation invariant. They are taken from Section~2.3 of~\cite{BF08}. Our cases are obtained by limits when all $a_i$'s goes to the same value, say $1$.
\begin{prop}[Proposition 2.8 of~\cite{BF08}]\label{propApp1}
Let $a_1,\ldots,a_n$ be non-zero complex numbers and let $F(z)$ be an analytic function in an annulus $A$ centered at the origin that contains all $a_j^{-1}$'s. Assume that $F(a_j^{-1})\neq 0$ for all $j$. Then, for $x^n\in W_n$,
\begin{equation}
\frac{\sum_{y^n\in W_n}\det\big(a_i^{y_j^n}\big)_{i,j=1}^n \det(f(x_i^n-y_j^n))_{i,j=1}^n}{F(a_1^{-1})\cdots F(a_n^{-1})} = \det\big(a_i^{x_j^n}\big)_{i,j=1}^n,
\end{equation}
where
\begin{equation}
f(m)=\frac{1}{2\pi\I}\oint_{\Gamma_0}\D z \frac{F(z)}{z^{m+1}}.
\end{equation}
\end{prop}

A simple corollary for the specific case of the transition probability in Definition~\ref{DefDiscreteCharlier} is the following.
\begin{cor}
For $F(z)=1-p+pz^{-1}$, it holds
\begin{equation}
f(m)=\frac{1}{2\pi\I}\oint_{\Gamma_0}\D z \frac{F(z)}{z^{m+1}}
=\left\{
         \begin{array}{ll}
           p, & \textrm{if }m=-1, \\
           1-p, & \textrm{if }m=0, \\
           0, & \textrm{otherwise},
         \end{array}
       \right.
\end{equation}
and
\begin{equation}
\sum_{y^n\in W_n}\Delta_n(y^n) \det(f(x_i^n-y_j^n))_{i,j=1}^n=\Delta_n(x^n).
\end{equation}
\end{cor}

\begin{prop}[Proposition 2.9 of~\cite{BF08}]\label{propApp2}
Let $a_1,\ldots,a_n$ be non-zero complex numbers and let $F(z)$ be an analytic function in an annulus $A$ centered at the origin that contains all $a_j^{-1}$ for $j=1,\ldots,n-1$ and assume that $F(a_j^{-1})\neq 0$ for all $j=1,\ldots,n-1$. Let us set $y_n^{n-1}=\virt$ and $f(x-\virt)=a_n^x$. Then
\begin{equation}
\frac{\sum_{y^{n-1}\in W_{n-1}}\det\big(a_i^{y_j^{n-1}}\big)_{i,j=1}^{n-1} \det(f(x_i^n-y_j^{n-1}))_{i,j=1}^n}{F(a_1^{-1})\cdots F(a_{n-1}^{-1})} = \det\big(a_i^{x_j^n}\big)_{i,j=1}^n.
\end{equation}
\end{prop}

A corollary concerning the transition kernel (\ref{eqMarkovLink}) is the following.
\begin{cor}\label{CorA4}
Let us choose $F(z)=(1-z)^{-1}$. Then
\begin{equation}
f(m)=\frac{1}{2\pi\I}\oint_{\Gamma_0}\D z \frac{F(z)}{z^{m+1}}
=\left\{
         \begin{array}{ll}
           1, & \textrm{if }m\geq 0, \\
           0, & \textrm{otherwise},
         \end{array}
       \right.
\end{equation}
and, with $y_n^{n-1}=\virt$,
\begin{equation}
\sum_{y^{n-1}\in W_{n-1}}(n-1)!\Delta_{n-1}(y^{n-1}) \det(f(x_i^n-y_j^{n-1}))_{i,j=1}^n=\Delta_n(x^n).
\end{equation}
\end{cor}

By the above results, we can define the transition kernels
\begin{equation}
T_n(a_1,\ldots,a_n;F)(x^n,y^n) =\frac{\det\big(a_i^{y_j^n}\big)_{i,j=1}^n}{\det\big(a_i^{x_j^n}\big)_{i,j=1}^n}\frac{\det(f(x_i^n-y_j^n))_{i,j=1}^n}{F(a_1^{-1})\cdots F(a_n^{-1})}
\end{equation}
for $x^n,y^n\in W_n$, and
\begin{equation}
T^n_{n-1}(a_1,\ldots,a_n;F)(x^n,y^{n-1}) =\frac{\det\big(a_i^{y_j^{n-1}}\big)_{i,j=1}^{n-1} }{\det\big(a_i^{x_j^n}\big)_{i,j=1}^n}\frac{\det(f(x_i^n-y_j^{n-1}))_{i,j=1}^n}{F(a_1^{-1})\cdots F(a_{n-1}^{-1})}
\end{equation}
for $x^n\in W_n$ and $y^{n-1}\in W_{n-1}$. In our application, $T_n$ is $P_n$ and $T^n_{n-1}$ is $\Lambda^n_{n-1}$. The intertwining condition is then a consequence of the following result.
\begin{prop}[Proposition 2.10 of~\cite{BF08}]\label{propApp3}
Let $F_1$ and $F_2$ two functions holomorphic in an annulus centered at the origin and containing all the $a_j^{-1}$'s that are nonzero at these points.
Then,
\begin{equation}
\begin{aligned}
&T_n(F_1) T_n(F_2)=T_n(F_2) T_n(F_1)=T_n(F_1 F_2),\\
&T_n(F_1) T^n_{n-1}(F_2)=T^n_{n-1}(F_1) T_n(F_2)=T^n_{n-1}(F_1 F_2).
\end{aligned}
\end{equation}
\end{prop}


\begin{thebibliography}{10}

\bibitem{ACJM13}
M.~Adler, S.~Chhita, K.~Johansson, and P.~{van Moerbeke}, \emph{{Tacnode
  GUE-minor processes and double Aztec diamonds}}, Probab. Theory Relat. Fields
  \textbf{162} (2015), 275--325.

\bibitem{AJvM14}
M.~Adler, K.~Johansson, and P.~van Moerbeke, \emph{Double aztec diamonds and
  the tacnode process}, Adv. Math. \textbf{252} (2014), 518--571.

\bibitem{BBBCCV14}
D.~Betea, C.~Boutillier, J.~Bouttier, G.~Chapuy, S.~Corteel, and M.~Vuleti{\'
  c}, \emph{{Perfect sampling algorithm for Schur processes}}, arXiv:1407.3764
  (2014).

\bibitem{Bor10}
A.~Borodin, \emph{{Schur dynamics of the Schur processes}}, Adv. Math.
  \textbf{228} (2011), 2268–--2291.

\bibitem{BC11}
A.~Borodin and I.~Corwin, \emph{Macdonald processes}, Probab. Theory Relat.
  Fields \textbf{158} (2014), 225--400.

\bibitem{BF07}
A.~Borodin and P.L. Ferrari, \emph{{Large time asymptotics of growth models on
  space-like paths I: PushASEP}}, Electron. J. Probab. \textbf{13} (2008),
  1380--1418.

\bibitem{BF08}
A.~Borodin and P.L. Ferrari, \emph{Anisotropic growth of random surfaces in
  $2+1$ dimensions}, Comm. Math. Phys. \textbf{325} (2014), 603--684.

\bibitem{BFS07b}
A.~Borodin, P.L. Ferrari, and T.~Sasamoto, \emph{{Large time asymptotics of
  growth models on space-like paths II: PNG and parallel TASEP}}, Comm. Math.
  Phys. \textbf{283} (2008), 417--449.

\bibitem{RB04}
A.~Borodin and E.M. Rains, \emph{{Eynard-Mehta theorem, Schur process, and
  their Pfaffian analogs}}, J. Stat. Phys. \textbf{121} (2006), 291--317.

\bibitem{BBCCR15}
C.~Boutillier, J.~Bouttier, G.~Chapuy, S.~Corteel, and S.~Ramassamy,
  \emph{{Dimers on Rail Yard Graphs}}, arXiv:1504.05176 (2015).

\bibitem{CJY14}
S.~Chhita, K.~Johansson, and B.~Young, \emph{{Asymptotic domino statistics in
  the Aztec diamond}}, Ann. Appl. Probab. \textbf{25} (2015), 1232--1278.

\bibitem{CY14}
S.~Chhita and B.~Young, \emph{{Coupling functions for domino tilings of Aztec
  diamonds}}, Adv. Math. \textbf{259} (2014), 173--251.

\bibitem{CEP96}
H.~Cohn, N.~Elkies, and J.~Propp, \emph{Local statistics for random domino
  tilings of the {Aztec} diamond}, Duke Math. J. \textbf{85} (1996), 117--166.

\bibitem{DF90}
P.~Diaconis and J.A. Fill, \emph{Strong stationary times via a new form of
  duality}, Ann. Probab. \textbf{18} (1990), 1483--1522.

\bibitem{DM14}
E.~Duse and A.~Metcalfe, \emph{{Asymptotic geometry of discrete interlaced
  patterns: Part I}}, arXiv:1412.6653 (2014).

\bibitem{EKLP92}
N.~Elkies, G.~Kuperbert, M.~Larsen, and J.~Propp, \emph{{Alternating-Sign
  Matrices and Domino Tilings I and II}}, J. Algebr. Comb. \textbf{1} (1992),
  111--132.

\bibitem{FerAZTEC}
P.L. Ferrari, \emph{{Java animation of the shuffling algorithm of the Aztec
  diamong and its associated particles' dynamics (discrete time TASEP, parallel
  update)}}, {\small \verb+http://wt.iam.uni-bonn.de/ferrari/research/+\\
  \verb+animationaztec/AnimationAztec.html+} (2009).

\bibitem{FS06}
P.L. Ferrari and H.~Spohn, \emph{Domino tilings and the six-vertex model at its
  free fermion point}, J. Phys. A: Math. Gen. \textbf{39} (2006), 10297--10306.

\bibitem{JPS98}
W.~Jockush, J.~Propp, and P.~Shor, \emph{Random domino tilings and the arctic
  circle theorem}, arXiv:math.CO/9801068 (1998).

\bibitem{Jo03}
K.~Johansson, \emph{The arctic circle boundary and the {Airy} process}, Ann.
  Probab. \textbf{33} (2005), 1--30.

\bibitem{JN06}
K.~Johansson and E.~Nordenstam, \emph{{Eigenvalues of GUE minors}}, Electron.
  J. Probab. \textbf{11} (2006), 1342--1371.

\bibitem{Ken01}
R.~Kenyon, \emph{Dominos and {G}aussian free field}, Ann. Probab. \textbf{29}
  (2001), 1128--1137.

\bibitem{Nor08}
E.~Nordenstam, \emph{{On the Shuffling Algorithm for Domino Tilings}},
  Electron.~J.~Probab. \textbf{15} (2010), 75--95.

\bibitem{NY11}
E.~Nordenstam and B.~Young, \emph{{Shuffling on Novak Half-Hexagons and Aztec
  Half-Diamonds}}, Electr. J. Comb. \textbf{18} (2011).

\bibitem{OCon02}
N.~O'Connell, \emph{Random matrices, non-colliding processes and queues},
  Lecture Notes in Mathematics \textbf{1801} (2003).

\bibitem{OR01}
A.~Okounkov and N.~Reshetikhin, \emph{Correlation function of {S}chur process
  with application to local geometry of a random 3-dimensional {Y}oung
  diagram}, J. Amer. Math. Soc. \textbf{16} (2003), 581--603.

\bibitem{Pet12}
L.~Petrov, \emph{{Asymptotics of random lozenge tilings via Gelfand-Tsetlin
  schemes}}, Probab. Theory Relat. Fields \textbf{160} (2014), 429--487.

\bibitem{PS02}
M.~Pr{\"a}hofer and H.~Spohn, \emph{Scale invariance of the {PNG} droplet and
  the {A}iry process}, J. Stat. Phys. \textbf{108} (2002), 1071--1106.

\bibitem{Pro03}
J.~Propp, \emph{{Generalized Domino-Shuffling}}, Theoret. Comput. Sci.
  \textbf{303} (2003), 267--301.

\bibitem{TW06}
C.~Tracy and H.~Widom, \emph{{The Pearcey Process}}, Comm. Math. Phys.
  \textbf{263} (2006), 381--400.

\bibitem{ZJ00}
P.~Zinn-Justin, \emph{Six-vertex model with domain wall boundary conditions and
  one-matrix models}, Phys. Rev. E \textbf{62} (2000), 3411--3418.

\end{thebibliography}

\end{document}